\documentclass[11pt]{article}
\usepackage{amsmath,epsfig,natbib}

\newtheorem{thm}{Theorem}

\newtheorem{defi}{Definition}

\newenvironment{proof}[1][Proof]{\begin{trivlist}
\item[\hskip \labelsep {\bfseries #1}]}{\end{trivlist}}

\newcommand{\captionfonts}{\small}

\makeatletter  
\long\def\@makecaption#1#2{%
  \vskip\abovecaptionskip
  \sbox\@tempboxa{{\captionfonts #1: #2}}%
  \ifdim \wd\@tempboxa >\hsize
    {\captionfonts #1: #2\par}
  \else
    \hbox to\hsize{\hfil\box\@tempboxa\hfil}%
  \fi
  \vskip\belowcaptionskip}
\makeatother   

\begin{document}
\title{Parallel hierarchical sampling: a practical multiple-chains sampler for Bayesian model selection}
\author{Fabio Rigat\footnote{Research fellow, CRiSM, Department of
    Statistics, University of Warwick; f.rigat@warwick.ac.uk}}
\date{May 14th, 2007}
\maketitle
\begin{abstract}
This paper introduces the parallel hierarchical sampler (PHS), a
Markov chain Monte Carlo algorithm using several chains simultaneously.
The connections between PHS and the  parallel tempering (PT) algorithm
are illustrated, convergence of PHS joint transition kernel is proved 
and and its practical advantages are emphasized. We  illustrate  the  inferences
obtained using PHS, parallel tempering and the
Metropolis-Hastings  algorithm  for  three  Bayesian  model  selection
problems, namely Gaussian clustering,  the selection of covariates for
a  linear regression model  and the  selection of  the structure  of a
treed survival  model.
\end{abstract}
\section*{Introduction}
Let  $\theta \in \Theta$  be  a random  variable  with distribution
$\Pi(\theta)$.   Markov  chain  Monte  Carlo (MCMC) 
algorithms generate discrete-time Markov chains  $\{\theta_{i}\}_{i=1}^{N}$
having $\Pi(\theta)$ as their unique stationary distribution
(\cite{RobCas}). 
MCMC  methods were pioneered  by  \cite{Ulam}  and  by
\cite{Metropolis}  in  the  field  of  statistical  mechanics.
They have been adopted in statistics  to approximate numerically 
expectations  of   the  form  $E_{\Pi}(g(\theta))$  where
$g(\cdot)\in  L^{2}(\Pi)$,  i.e.  the  function $g(\cdot)$  is  square
integrable with  respect to $\Pi(\theta)$. In Bayesian
statistics, when the posterior distribution of a
parameter $\theta$ given the data $X$, $\Pi(\theta  \mid  X)$, cannot
be integrated analytically with  respect to its dominating measure,
its relevant features can be approximated
via MCMC (\cite{Tierney}). For  a thorough analysis  of the  published MCMC
algorithms, the  reader may refer  to \cite{GelfSmith}, \cite{SR}, 
\cite{NealTRep}, \cite{Gilks}, \cite{Dani}, \cite{RobCas} 
and \cite{Liu} among others.

This paper illustrates  a novel Markov chain algorithm,  which we find
useful for  sampling from highly multimodal  target distributions.  We
label this  algorithm parallel  hierarchical sampler (PHS)  because of
the prominent  role of one chain  with respect to  the other generated
chains. An important feature of PHS is that one array of Monte Carlo
samples is generated using many chains run in parallel. PHS has in
fact many connections with the Metropolis-coupled Markov chain Monte Carlo  
 samplers of \cite{SWang}, \cite{Geyer} and of \cite{HN}. The main
  advantage of PHS with respect to other samplers is that 
the proposed updates are always accepted, thus ensuring optimal
mixing of the resulting chain. 

In Section  $1$ of this paper we review some foundations of 
MCMC   methods  relevant   for  our work,  with   emphasis   on  the
Metropolis-Hastings (MH) algorithm and on parallel tempering (PT).  In
Section  $2$  we  introduce  the   PHS  algorithm,  we  prove  the
reversibility of its joint transition kernel with respect to its
target  distribution and we illustrate the relationships between PT
and PHS.   Sections $3$ and  $4$ illustrate  two
examples  comparing the  inferences obtained  using the  PHS algorithm
with  those  of  MH  and   PT  within  the  Bayesian  model  selection
framework. The  first example deals  with data clustering.  In the
second example we consider the problem of selecting the best subset of
covariates for a Gaussian  linear regression model.  Section $5$
illustrates  the application  of PHS  for deriving posterior
inferences for the  structure of a treed survival  model.  Section $6$
discusses the  current results and  some ongoing developments  of this
work.
\section{Foundations of MCMC algorithms}
Let $\theta$ be a random variable with  distribution  $\Pi(\theta)$.
In what follows, if f $\theta$  is continuous we let $f(\theta)$ be its probability  density
 with  respect to Lebesgue measure whereas if  it is
discrete $f(\theta)$ is  its probability mass function.
When it is not possible to obtain independent draws from $\Pi(\theta)$, 
Markov chains can be used to generate dependent realisations $\{\theta_{i}\}_{i=1}^{N}$ 
having stationary distribution  $\Pi(\theta)$. Here  
the conditions under which such Markov chains can be
constructed are stated and two algorithms are illustrated.

Let $K(\theta_{i},\theta_{i+1})$  be a transition kernel 
defining the probability to jump between any two values of $\theta$. 
If there exist an integer $d>0$ such that the probability 
of a transition between any two values $(\theta_{i},\theta_{i+1})$
with $f(\theta_{i})>0$ and $f(\theta_{i+1})>0$
in $d$ steps is positive, the transition kernel is $\Pi$-irreducible.
$K(\theta_{i},\theta_{i+1})$ is aperiodic if
it does not entail cycles of transitions among states.
 A sufficient condition for aperiodicity of
an irreducible transition kernel is that $K(\theta_{i},\theta_{i})>0$
for some $\theta_{i} \in \Theta$.
The transition kernel is reversible  with  respect to  $\Pi(\theta)$ if
it satisfies the detailed balance (DB) condition
\begin{eqnarray}
f(\theta_{i})K(\theta_{i},\theta_{i+1}) = f(\theta_{i+1})K(\theta_{i+1},\theta_{i}).
\end{eqnarray} 
If a reversible and $\Pi$-irreducible transition kernel is also aperiodic,
$\Pi(\theta)$ is its unique stationary distribution
(\cite{Nummelin}, \cite{RobCas}).   In such  case,  the  strong law of
large numbers  holds for any  function $g(\cdot) \in  L^{2}(\Pi)$,
that is
$$
\lim_{N \to \infty} \left(\frac{1}{N}\sum_{i=1}^{N}g(\theta_{i})\right)=E_{\Pi}(g(\theta))\mbox{ a.s. }
$$
Furthermore, under these  conditions also the central limit theorem
holds so that (\cite{Tierney})
$$
\sqrt{N}\left (\frac{1}{N}\sum_{i=1}^{N}g(\theta_{i})-E_{\Pi}(g(\theta))\right )\stackrel{d}{\rightarrow} N(0,\sigma_{g,K}^{2}),
$$
where the  asymptotic standard deviation
$\sigma_{g,K}$  depends on  the  function  $g(\cdot)$   and  on  the
transition  kernel  (\cite{Mira3}).
\subsection{The Metropolis-Hastings algorithm}
The Metropolis-Hastings algorithm (\cite{MH})
implements a family of transition kernels reversible with respect to
an arbitrary target distribution $\Pi(\theta)$.
Markov chains are generated by the MH algorithm
by converting  the independent  draws from
a proposal  distribution $q(\cdot)$  into dependent  samples from
$\Pi(\theta)$ through a  simple  accept/reject mechanism
(\cite{Greenberg}, \cite{BilleraDiac}). Without  loss of generality, in  what follows we
let $q(\cdot)$ assign zero probability to the current state $\theta_{i}$. Under this
condition, the MH transition kernel can be written as
\begin{eqnarray}
K_{MH}(\theta_{i},\theta_{i+1}) = \alpha_{MH}(\theta_{i},\theta_{i+1})q(\theta_{i+1}\mid \theta_{i})\text{ if }\theta_{i+1}\neq \theta_{i},
\end{eqnarray}
where the acceptance ratio $\alpha_{MH}(\theta_{i},\theta_{i+1})$ is
defined as 
\begin{eqnarray}
\alpha_{MH}(\theta_{i},\theta_{i+1}) = 1 \wedge \frac{f(\theta_{i+1})q(\theta_{i}\mid \theta_{i+1})}{f(\theta_{i})q(\theta_{i+1}\mid \theta_{i})}.
\end{eqnarray}

The irreducibility  and aperiodicity of
the  MH  transition  kernel and its practical performance for specific
sampling problems hinge mainly  on the appropriate choice  of  the  proposal
distribution  $q(\cdot)$ (\cite{Tierney}). Reversibility with respect to $f(\theta)$
can be immediately verified by checking  that $(1)$ holds using equations $(2)$ and $(3)$.
Furthermore, by equation $(3)$ the MH  algorithm requires the evaluation
of the function
$f(\theta)$ only  up to a multiplicative constant, which is a
key feature for Bayesian applications where the target posterior distribution
is typically known up to a finite multiplicative factor (\cite{GelfSmith}).
\subsection{Parallel tempering}
Parallel tempering (PT) is a multiple-chains  extension of the
MH algorithm  which can improve mixing 
when the MH sample  paths exhibit  high correlations  (\cite{Geyer}).
Poor mixing of the MH algorithm typically arises
when  the target distribution  is multimodal. As  emphasized by
\cite{SWang}, \cite{HN} and
\cite{Liu}, in  statistical mechanics such distributions  may arise in
the analysis  of stochastic interacting particle systems,  such as spin glasses
and the Ising model. In Bayesian statistics multimodality of the
posterior distribution might  occur when an informative prior disagrees with the likelihood,  
in highly structured hierarchical models and when 
several nuisance parameters are integrated out 
of the joint posterior. An impotant example of the latter case
is provided by Bayesian model  selection methods using the marginal posterior
probability  of the model  structure.

The PT algorithm appears in the literature in different forms and under different
labels. \cite{SWang} introduced the ``Replica Monte Carlo'' algorithm,
\cite{Geyer} labeled his algorithm ``Metropolis-coupled MCMC'', whereas
the algorithm of \cite{HN} is labeled ``Exchange Monte Carlo''.
Here we give a unified description of PT
encompassing those of Geyer, Liu and Hukushima and Nemoto.

Let $1  = T_{1} \leq T_{2} \leq ...\leq T_{M} <  \infty$ be a fixed real-valued vector
of temperature levels.  For each value of the index $m \in [1,M]$ a
heated version of the target density $f(\theta)$ is defined by
``powering up'' $f(\theta)$ as
\begin{eqnarray}
f_{m}(\theta) = \frac{f(\theta)^{\frac{1}{T_{m}}}}{C_{m}},
\end{eqnarray}
where $C_{m}  > 0$ is a  finite normalising constant depending  on the temperature
parameter  $T_{m}$.  The latter  acts  as  a  smoother of  the  target
distribution  of the  cold chain, which has temperature one, so that  the heated  densities have
fatter   tails   and   less   pronounced   modes   with   respect   to
$f_{1}(\theta)$.  The PT sampler  proceeds, at  each iteration, by alternating
an \emph{update} step with a \emph{swap} step. The former is carried out by updating
each chain independently of the others typically via the Gibbs sampler using one or more embedded MH steps.
To perform the \emph{swap} step, let $s_{i}$ have value $0$ if \emph{update}
is chosen at iteration $i$ and $1$ if \emph{swap} is chosen instead.
The proposal probability $q^{'}_{s}(s_{i} \mid s_{i-1})$ describes how
the two steps are combined by the sampler. \cite{Geyer}
adopts the deterministic proposal $q^{'}_{s}(s_{i} \mid s_{i-1}) =
1_{\{s_{i-1} = 0 \}}$, whereas \cite{Liu} defines an independent PT sampler using
$q^{'}_{s}(s_{i} \mid s_{i-1}) = s$ where $s \in (0,1)$ is a fixed
swap proposal rate.
Let the indexes $j_{i}$ and $k_{i}$ range over $(1,...,M)$ and let
$\theta_{i}^{j_{i}}$ indicate the state of chain $j_{i}$ at iteration $i$.
The second proposal employed by the PT sampler, $q^{''}_{s}(\theta_{i}^{j_{i}},\theta_{i}^{k_{i}})$
defines the probability that at iteration $i$ a swap is attempted
between the current values of the chains with indexes $(j_{i},k_{i})$.
In \cite{Geyer}, in \cite{HN} and in \cite{Liu} this proposal
is taken as independent of the current states of the two chains
$\theta_{i}^{j_{i}},\theta_{i}^{k_{i}}$ but only
dependent on thei indexes $(j_{i},k_{i})$. Specifically, the swap
proposal used by these authors is uniform over all possible values of the ordered couple $(j_{i},k_{i})$
with $k_{i} \neq j_{i}$ and a swap is accepted
with probability
\begin{eqnarray}
&\alpha_{s}([\theta_{i}^{j_{i}},\theta_{i}^{k_{i}}],[\theta_{i}^{k_{i}},\theta_{i}^{j_{i}}]) = 1 \wedge\frac{f_{j_{i}}(\theta_{i}^{k_{i}})f_{k_{i}}(\theta_{i}^{j_{i}})}{f_{j_{i}}(\theta_{i}^{j_{i}})f_{k_{i}}(\theta_{i}^{k_{i}})},
\end{eqnarray}
ensuring the reversibility of the PT sampler with respect to its joint target distribution.
When the independent updates of each chain are carried out using a single MH step,
the joint transition kernel of the PT sampler is
\begin{eqnarray}
K_{PT}(\theta_{M,i},\theta_{M,i+1}) =
(1-q^{'}_{s}(s_{i} \mid s_{i-1}))\prod_{w=1}^{M}q(\theta^{w}_{i+1} \mid \theta^{w}_{i})\alpha_{MH}(\theta^{w}_{i},\theta^{w}_{i+1}) + \nonumber \\
 + q^{'}_{s}(s_{i} \mid s_{i-1}) \sum_{j_{1} = 1}^{M}\sum_{\stackrel{k_{i} = 1}{k_{i} \neq j_{i}}}^{M}q^{''}_{s}(j_{i},k_{i})
\alpha_{s}([\theta_{i}^{j_{i}},\theta_{i}^{k_{i}}],[\theta_{i}^{k_{i}},\theta_{i}^{j_{i}}]).
\end{eqnarray}
 where
$\theta_{M,i} = [\theta_{i}^{1},...,\theta_{i}^{M}]$ is the state of all $M$ chains at iteration $i$.
From $(6)$ it can be seen that when the within-chain  updates  produce
high correlations, PT increases mixing for all chains through their successful
swaps.  Analogously to the MH  algorithm, the irreducibility
and aperiodicity  of the PT  transition kernel depend mainly on  the proposal
distribution for the within-chains update, $q(\cdot)$ and on that of
the cross-chains swaps $q^{''}_{s}(\cdot)$.  
A proof of the reversibility of the PT algorithm
can  be found  in \cite{HN}.  Finally  we note  that,
analogously  to the  MH algorithm,  by equations  $(4)$ and  $(5)$ the
implementation  of PT does  not require  knowledge of  the finite normalising
constants $\{C_{m}\}_{m=1}^{M}$ so that it is a suitable MCMC sampler
for Bayesian posterior simulation.
\section{The parallel hierarchical sampler}
A key difficulty affecting the general applicability of the PT sampler is 
its dependence on the values of the temperatures
$\{T_{m}\}_{m=1}^{M}$. In statistical mechanics, the latter are chosen
with reference to the physical properties of the systems being
modeled, such as the energy barriers implied by successive temperature
levels. However, in statistics the equilibrium distributions being
simulated seldom possess analogous interpretations. An alternative
solution illustrated in this Section is to employ a multiple chains 
sampler such that the equilibrium distributions of all chains is the 
same but the proposal distribution used to update each chain is different. 
Specifically, definition $1$ describes a multiple-chains 
sampler which does not employ temperatures and combines independent
updates with swap moves within each iteration.
\begin{defi}
Let a multiple-chains MCMC sampler proceed
by carrying out both the following two steps at each iteration:
 \begin{itemize}
\item[i)] let the index $m_{i}$ be drawn from a discrete proposal distribution
$q^{''}_{s}(m_{i} \mid m_{i-1})$ symmetric with respect to its arguments;
\item[ii)] swap the current value of chain $m_{i}$ and that of the first chain;
\item[iii)] update independently the remaining $M-2$ chains each having
  the same marginal target distribution $f(\theta)$.
\end{itemize}
\end{defi}
At point $ii)$ above, we indicate as $q^{''}_{s}(\cdot)$ the swap
proposal to emphasize the analogy with the PT algorithm.
We label the algorithm defined above parallel hierarchical sampler (PHS) because 
the first chain is given a prominent role and the update of all chains 
is carried out in parallel analogously to PT. To provide a simple
proof of the reversibility of the PHS joint kernel,
in this Section we assume that the chains $(2,..M)$ are updated
using a single MH step and that the transition  kernels
for these MH updates satisfy the conditions
illustrated in \cite{Tierney} so that they
are irreducible and aperiodic with respect to  their marginal target
distributions. In addition, we assume that the symmetric proposal 
distribution $q^{''}_{s}(\cdot)$ allows for swaps between the first
chain and any of the other chains. Under these conditions the marginal transition kernel for the first chain of the PHS algorithm
is irreducible and aperiodic with respect to its target distribution.
Let $\theta_{M} = \theta \times \theta \times...\times \theta$ be the $M$-fold cartesian product of
the random variable $\theta$. By the arguments of Section $1$, if the PHS joint transition
kernel is also reversible with respect to the product density
$\mu(\theta_{M})$ having all marginals equal to $f(\theta)$, then 
$\mu(\theta_{M})$  is the unique joint stationary distribution of the sampler.
The reversibility of  the PHS is proved  in the following theorem.
\begin{thm}
The joint transition kernel  of the PHS algorithm of Definition $1$ is reversible
with respect to the joint distribution having product density or probability mass function $\mu(\theta_{M})$.
\end{thm}
\begin{proof}
The DB condition for the PHS algorithm is
\begin{eqnarray}
\frac{\mu(\theta_{M,i})}{\mu(\theta_{M,i+1})}=
\frac{K_{PHS}(\theta_{M,i+1},\theta_{M,i})}{K_{PHS}(\theta_{M,i},\theta_{M,i+1})},
\end{eqnarray}
where $K_{PHS}(\theta_{M,i+1},\theta_{M,i})$ is the PHS joint transition kernel.
When the independent updates of the chains $(2,...,M)$
are carried out via a MH step, the PHS joint transition kernel can be
written explicitely as
\begin{eqnarray}
K_{PHS}(\theta_{M,i},\theta_{M,i+1}) =
\sum_{m_{i}=2}^{M}q^{''}_{s}(m_{i} \mid m_{i-1}) \prod_{\stackrel{j =2}{j \neq m_{i}}}^{M}q(\theta^{j}_{i+1} \mid \theta^{j}_{i})\alpha_{MH}(\theta_{i}^{j},\theta_{i+1}^{j}).
\end{eqnarray}
Each summand in $(8)$ is the product of the marginal transition kernel
for the swap transition and those of the $(M-2)$ independent MH
updates for the remaining chains. The former 
coincides with the proposal $q^{''}_{s}(m_{i} \mid m_{i-1})$
because the PHS swap acceptance ratio is equal to one.
This fact will be motivated in the next Section by illustrating the
relationship between PHS and PT. 
Under $(8)$ the DB condition $(7)$ can be rewritten as
\begin{eqnarray}
\sum_{m_{i}=2}^{M}q^{''}_{s}(m_{i} \mid m_{i-1}) \prod_{\stackrel{j =2}{j \neq m_{i}}}^{M}q(\theta^{j}_{i+1} \mid
\theta^{j}_{i})\alpha_{MH}(\theta_{i}^{j},\theta_{i+1}^{j}) =
    \nonumber \\
= \sum_{m_{i}=2}^{M}q^{''}_{s}(m_{i-1} \mid m_{i}) \prod_{\stackrel{j =2}{j \neq m_{i-1}}}^{M}q(\theta^{j}_{i} \mid
\theta^{j}_{i+1})\alpha_{MH}(\theta_{i+1}^{j},\theta_{i}^{j})
\end{eqnarray}
For any given value of $m_{i}$, by the reversibility of $(2)$ and
$(3)$ with respect to $f(\theta)$, the $M-2$ MH transition probabilities
on the left-hand side of $(9)$ are equal to their
corresponding terms on the right-hand side. By taking
$q^{''}_{s}(\cdot)$ symmetric with respect to $m_{i}$ and $m_{i-1}$, 
for all values of $m_{i}$ each summand on the left-hand side of $(9)$
equals its corresponding term on the right-hand side, so that
the equality $(9)$ holds.\hspace{0.5in}$\diamond$
\end{proof}
Equation $(8)$ implies that, as for the MH and PT algorithms, PHS
does   not   require    knowledge   of   the normalising  constant
of its marginal target distributions $C$ so that it is suitable for sampling from target
distributions known only up to a finite multiplicative factor.
\subsection{Relationship between PHS and parallel tempering}
Both $(6)$ and $(8)$ are mixtures of marginal transition kernels
respectively defining the joint transition probabilities for the PT
and PHS algorithms. The analogy between the two is that
$(M-1)$ out of the $M$ parallel chains are auxilliary and Monte Carlo estimates
are computed using the samples of the first chain only. 
There are two important differences between the two samplers. 
At each iteration, the PHS transition kernel mixes over the update and
swap steps as described in Definition $1$ whereas in PT they
are alternated according to the proposal probability $q^{'}_{s}(s_{i}
\mid s_{i-1})$. Since the former step typically generates local transitions
whereas the latter produces larger jumps, PT
creates unnecessary competition between local and global
mixing. Furthermore, in PHS all marginal
target distributions are not powered up using a temperature
coefficient as in PT. The rationales to avoid the
temperature coefficients are both conceptual and practical.
From a Bayesian prespective, the main conceptual issue is
that the temperatures do not appear neither in the likelihood function
nor in the prior, so that it is not clear whether they should be treated analogously to the
other parameters indexing the target posterior distribution. In practice,
determining sensible values for the temperatures requires a lenghty
trial-and-error process in the pursuit of a target swap rate between
pairs of chains. Moreover,
since the normalising constants of the marginal posterior
distributions depend on their temperatures, updating of the
latter is not possible unless for conjugate
families. The simulated tempering algorithm of \cite{GThompson} 
implements a single chain sampler for both the parameter $\theta$ and
a single temperature coefficient $T$. The latter is
treated as a discrete random variable and its update is 
carried out using a data dependent pseudo-prior in order to simplify
the normalising constant of the joint posterior distribution from the
Metropolis-Hastings acceptance ratio. An interesting point in
simulated tempering is that the temperature is not constrained to be larger than
one, so that when $T$ is close to zero the posterior distribution
becomes concentrated around its modes. However, although practically useful, the
simulated tempering algorithm does not clarify the nature
of the temperature parameter and it does not explain how
the posterior normalising constant could be seen as part of a prior
distribution for the same parameter. 

In PHS, since all temperatures have value $1$, the Metropolis swap
acceptance ratio $(5)$ is equal to one, so that the proposed moves
for the first chain are always accepted. This property marks the most
evident difference between the sample paths of the first chain of PHS,
those of the cold chain of PT and those of the MH algorithm. 
\subsection{PHS as a variable augmentation scheme}
Variable augmentation for MCMC samplers was first introduced by
\cite{TWong}. Its general principle is that convergence of one of the
generated chains can be sped up by cleverly augmenting the
state-space using additional coefficients. Conditionally on these
auxilliary variables the posterior distribution of the parameters of interest can typically
be sampled exactly. The PHS algorithm can be seen as a 
variable augmentation scheme where the additional coefficients are 
$M-1$ replicates of the parameter of interest itself. In PHS
the target distribution of the first chain does not 
depend on the other replicates. At each iteration, the $M-1$
auxilliary chains having index $(2,..,M)$ directly provide a set of
potential updates for the chain of interest.
\subsection{PHS and multiple-try Metropolis algorithms}
\cite{mtry} illustrate a generalization of the single chain Metropolis
algorithm where multiple values from the same proposal distribution
are drawn at each iteration of the sampler. To attain detailed
balance, a generalized Metropolis acceptance ratio involving several
pseudo-current chain states is computed. This multiple-try generalized 
Metropolis sampler actually mimics the behaviour of a multiple chains
algorithm using the sme proposal within each of the generated chains.
Therefore, the main analogy between the algorithm of \cite{mtry} and
PHS is that many candidates are available at each iteration to
 update one chain of interest. In \cite{mtry}, only one of such
 updates is retained and the Metropolis ratio is modified
 accordingly. In PHS, all such values not used for swapping with the
 first chain are retained and individually updated. Moreover, the
 proposal mechanism generating all potential updates is not
 constrained to be the same for all chains.
 \section{An illustrative example: MCMC generation of mixtures of Gaussian variates}
In this Section we report a comparison between the empirical
performance of MH and PHS algorithms for generating a sample from
a mixture of scalar Gaussian random variables. We use the results of
one simulation to illustrate the typical difference between the
performance of the two samplers.

Within the MH algorithm we use a random walk uniform proposal
distribution on the interval $(\theta_{i}-\delta,\theta_{i}+\delta)$,
where $\theta_{i}$ is the current value of the chain. We construct a PHS
algorithm using the same Metropolis updates within chains $(2,...,M)$
and by adopting a uniform swap proposal distribution $q^{''}(m_{i}
\mid m_{i-1}) = \frac{1}{M-1}$. For this example we let $M = 10$, that is we employ
nine auxilliary chains having the same proposal spread
$\delta = 1$ as that of the MH sampler. The PHS sampler was run for one
hundred thousand iterations. In order to make the computational cost
for both samplers comparable, the MH algorithm was run for one million iterations.
All chains were started at the same initial value equal to zero.

The number of components of the mixture was set to $5$, 
their means were generated uniformly at random over the interval $(-10,10)$
obtaining the values $(-8.85,-2.65,2.63,3.85,4.35)$. Their standard
deviations were generated uniformly at random over the interval
$(0.1,1)$, obtaining the values $(0.18, 0.51, 0.50,0.42,0.24)$.
Finally, the unnormalised weights of the mixture components were drawn uniformly at
random over the interval $(1,5)$. Thier normalised values are $(0.22,0.22,0.23,0.15,0.18)$.
Figure $1$ shows the probability density of the mixture over the range
$(-13.5,6.6)$. The mixture components having means $(2.63,3.85)$ and
standard deviations $(0.50,0.42)$ are very close and they do not
result in two separate modes of the mixture density. Figure $2$ compares the histograms
of the MH draws with that of the first PHS chain. The former sampler
effectively located the three closest modes to its starting value
whereas PHS successfully visited all four modes of its target
distribution. 
\begin{center}
  \begin{figure}[htbp]
\hspace{1in}
   \epsfig{file=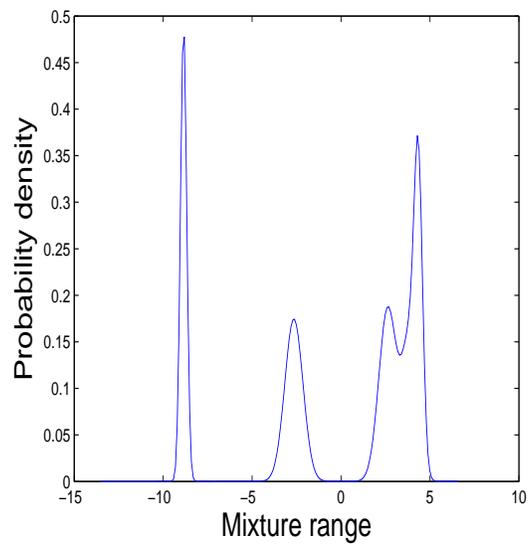,width=3in,height = 3in,angle=0}
    \caption{the Gaussian mixture density used as target distribution
    for the Metropolis sampler and for the parallel hierarchical
    sampler. The distribution is a mixture of five Gaussian components
    having means $(-8.85,-2.65,2.63,3.85,4.35)$, standard deviations
    $(0.18, 0.51, 0.50,0.42,0.24)$ and weights $(0.22,0.22,0.23,0.15,0.18)$.}
  \end{figure}
\end{center}
\begin{center}
  \begin{figure}[htbp]
   \epsfig{file=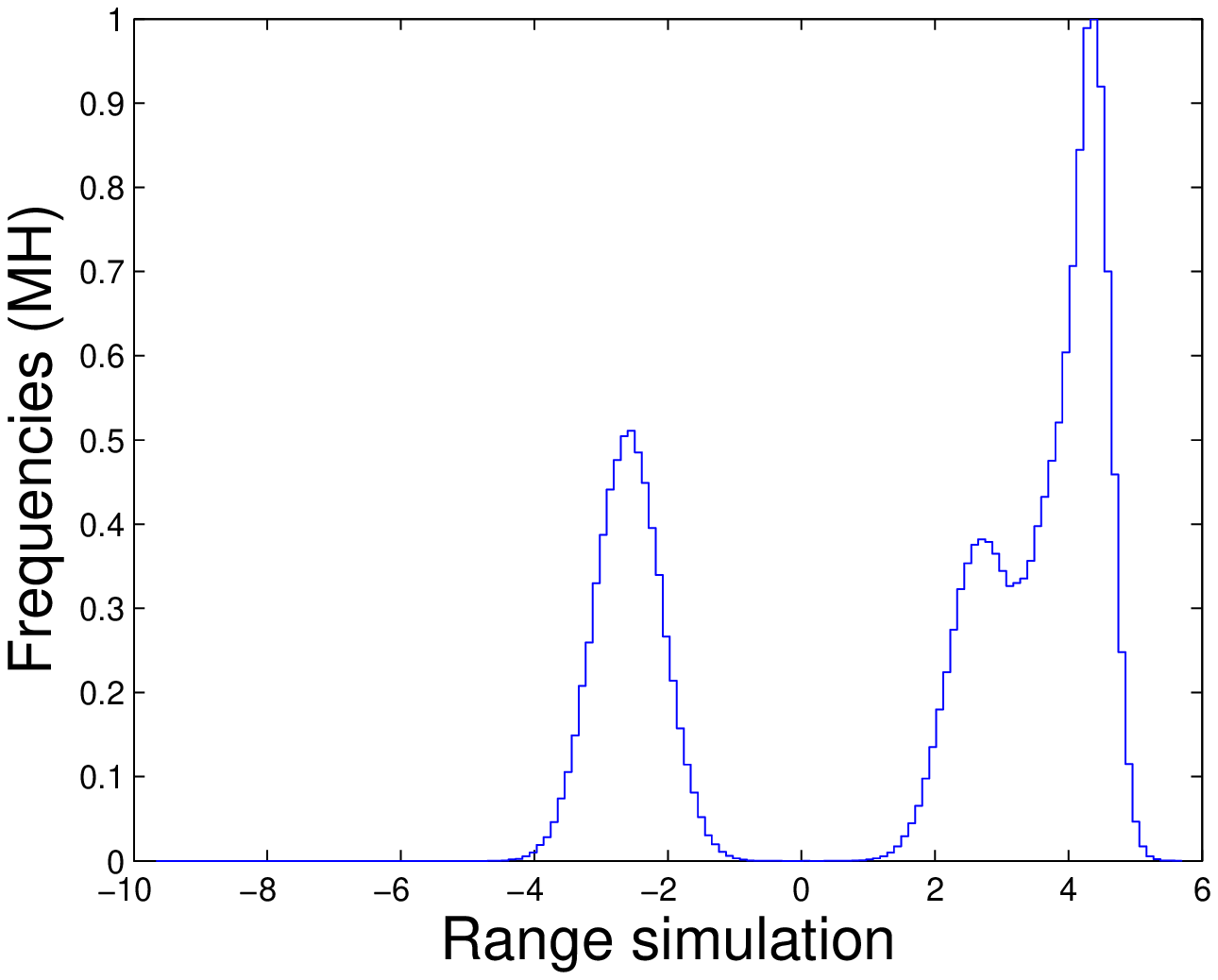,width=2.4in,height = 2.4in,angle=0}
   \epsfig{file=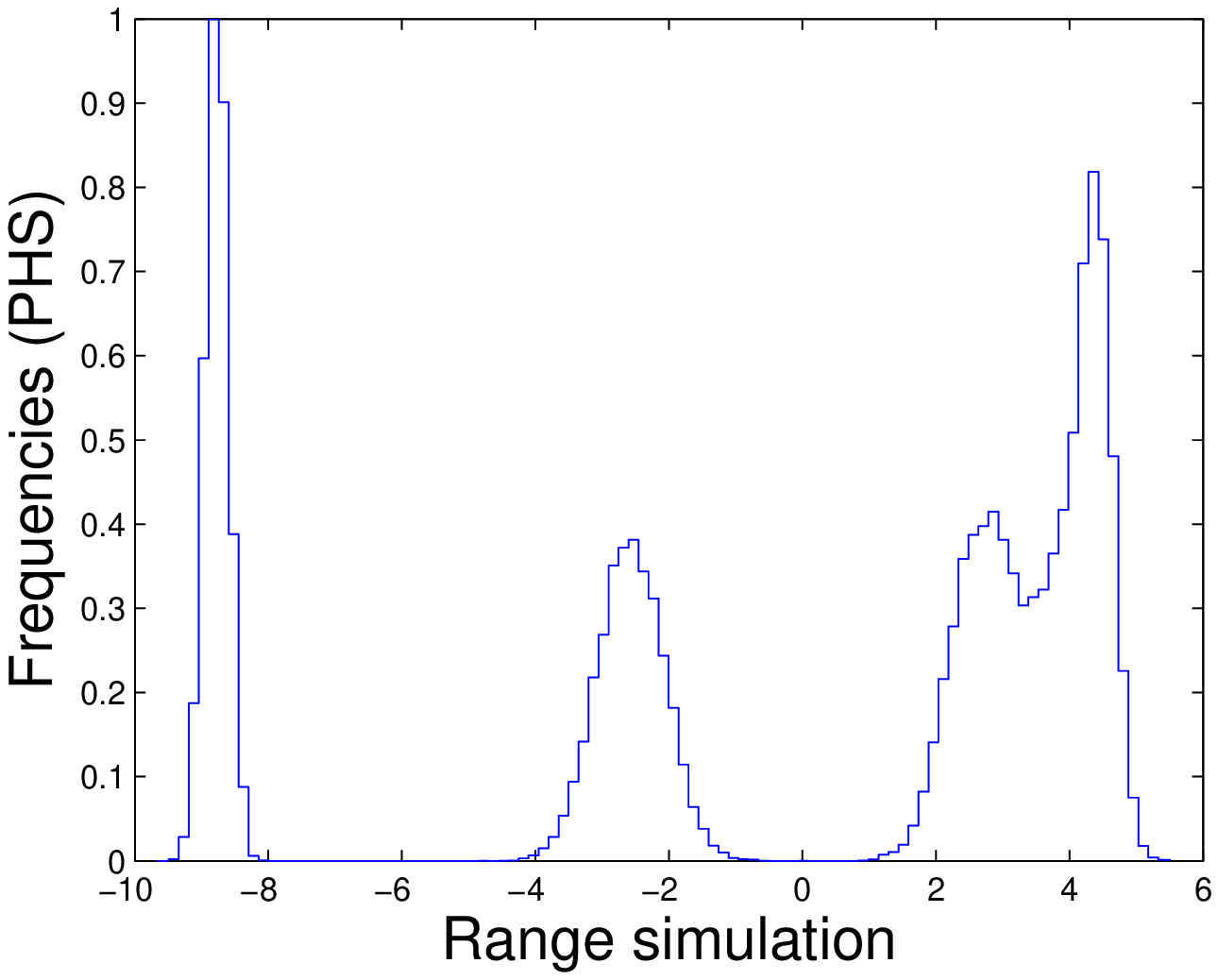,width=2.4in,height = 2.4in,angle=0}
    \caption{the plot on the left-hand side shows the histogram of the
    Metropolis draws. On the right-hand side, the plot represent the
    histogram of the draws of the first PHS chain. Thanks to the
    swapping mechanism, the latter successfully visited all the four
    modes of its marginal target distribution whereas the Metropolis
    algorithm only visited the three closes models to its starting value.}
  \end{figure}
\end{center}
\section{Application to the selection of covariates for the Bayesian linear regression model}
The covariates  selection problem for  the
Bayesian Gaussian linear regression  has been
addressed model using MCMC methods  by \cite{Mitchell}, \cite{SmithKohn}, \cite{SSVS},   \cite{CarlinChib95},
\cite{GeorgeMcCull97},       \cite{Raft97},       \cite{KuoMallick98},
\cite{DellapForster02}  and  \cite{clydegeorge04}  among many  others.

Using the same notation as in \cite{GeorgeMcCull97}, we let  the
distribution  of  the $n$-dimensional  random  vector $Y$  be
multivariate   Gaussian  with   mean   $X_{\gamma}\beta_{\gamma}$  and
covariance   matrix   $\sigma^{2}I_{n}$,   being
$(\sigma,\beta,\gamma)$  a priori  unknown. The  $p$-dimensional model
index $\gamma$ has elements $\gamma_{j}$ taking value one if the $j$th
covariate is  used for  the computation  of the mean  of $Y$  and zero
otherwise. Here $\beta_{\gamma}$ and $X_{\gamma}$ include respectively
the elements  of the $p$-dimensional column  vector $\beta$ associated
to non-zero  components of $\gamma$  and the corresponding  columns of
$X$.  The  latter is  a  real-valued $n  \times p$  matrix
representing $p$ potential predictors for the mean of $Y$.  Within this framework,
the  variable selection  problem consists  of deriving  inferences for
$\gamma$ conditionally on the data  $(Y,X)$.  In order to compute such
inferences, we employ MH, PT and PHS to  generate  draws  from  the marginal  posterior
probability of the model index,
 $$
 P(\gamma \mid Y,X) \propto P(\gamma)P(Y \mid \gamma,X).
 $$
In this
Section we adopt  the same form of the  marginal posterior probability
of $\gamma$ as in \cite{Nott}, letting
\begin{eqnarray}
 P(\gamma \mid Y,X) \propto (1+n)^{-\frac{S(\gamma)}{2}}\left(Y^{'}Y -\frac{n}{n+1}Y^{'}X_{\gamma}(X^{'}_{\gamma}X_{\gamma})^{-1}X^{'}_{\gamma}Y \right)^{-\frac{n}{2}}
\end{eqnarray}
We also note that if  the  predictors included  in
$X_{\gamma}$     tend     to      be     collinear,     the     matrix
$X^{'}_{\gamma}X_{\gamma}$ can  be almost singular and  the right hand
side  of  equation  $(11)$  may  be  numerically  unstable.   In  such
instances,  we find that  computing the  marginal posterior  using the
Cholesky   decomposition    of   $X^{'}_{\gamma}X_{\gamma}$,   as   in
\cite{SmithKohn}, yields numerically stable results.

In order  to compare the PHS estimates  with that of
the MH and of the PT algorithms,  we consider two simulated datasets.
In both cases the dependent data
is  $Y \sim N(X_{\gamma}\beta_{\gamma},6.25I_{180})$ and the regression
coefficients   are   set  at   $\beta_{\gamma_{j}} = 2j/15$ for $j = 1,...,15$.
For the first dataset, $X$ is generated as a
$180 \times 15$ matrix of \emph{i.i.d.} draws from a Normal distribution
with mean zero and variance $1$. Let  $Z_{1},...,Z_{16}$ be  $i.i.d$
Gaussian column vectors of length  $180$ with mean zero and covariance
matrix $I_{180}$. For the second dataset a strong collinearity was induced among
the predictors $X$ by letting
\begin{eqnarray*}
X_{j} &=& Z_{j} + 2Z_{16} \mbox{ for }j = 1,...,15,
\end{eqnarray*}
where $X_{j}$ is the $j$th column of $X$, as in Section $5.2.1$ of \cite{GeorgeMcCull97}. Here we will compare the estimation results
of the MH algorithm with those of the cold chain of PT and of PHS for the two
datasets using the estimated marginal posterior
inclusion probabilities for each predictor and their Monte Carlo standard  errors (MCSEs).
The former are defined as $\bar{\gamma_{j}} = \sum_{i=1}^{N}\gamma^{i}_{j}/N$
where $i = 1,...,N$ is the iteration index and $\gamma^{i}_{j}$
is the $i$th draw for the $j$th predictor. As illustrated by  \cite{Geweke}, \cite{Nott} and by \cite{GeorgeMcCull97},
the MCSE  for the inclusion probability of the $j$th predictor is
$$
MCSE(\bar{\gamma}_{j}) = \sqrt{\frac{1}{N}\sum_{|h|<N}\left(1-\frac{|h|}{N} \right)A_{j}(h)},
$$
where $A_{j}(h)$ is the lag $h$ autocovariance of the chain of realisations for $\gamma_{j}$. For
ergodic Markov chains, as $N \rightarrow \infty$ the  MCSE converges, up to
an additive  constant independent of the transition kernel,
to the MCMC standard  error $\sigma_{g,K}$ (\cite{Mira3}) where
$g(\gamma_{j}) = E(\gamma_{j} \mid Y,X)$ for this example.

Three independent  batches of  chains were run for  fifty thousand
iterations.  For PT and PHS, we used nine  chains which target  distributions are  defined as in
$(4)$ with  cold distribution  $(11)$. For PT, the heated chains were defined using the same  array of  equally spaced
temperatures with range $1-5$. For each sampler, the starting values of $\gamma$ for
all  chains was the null model.  All within-chain updates  were  carried  out using  a
component-wise random  scan Metropolis algorithm proposing a change of
the current  value of each parameter $\gamma_{j}$ at
every  iteration,  as  in \cite{Denison98}. The cross-chains proposals $q_{s}(\cdot)$ and $q^{''}_{s}(\cdot)$ were taken uniform. Since the PT algorithm
also depends on the proposal $q^{'}_{s}(\cdot)$, we run three batches of PT chains
using Liu's proposal $q^{'}_{s}(s_{i} \mid s_{i-1}) = s$ with $s = 0.2,0.5,0.8$.

Figure $5$ illustrates the simulation results. The plots on the top row refer to the data
without collinearity whereas the plots on the bottom report the inferences for
the data with collinearity. By comparing the two rows of Figure $5$,
it appears that the induced collinearity among the predictors did not affect the estimation results
for any of the samplers. The estimated inclusion probabilities are generally
increasing with respect to the true value of their regression parameters $\beta_{\gamma}$ for all samplers
and for both datasets, with a noticeable shift occurring between the fifth and the sixth predictors
(which regression parameters are respectively $\beta_{\gamma_{5}} = 0.67$ and $\beta_{\gamma_{6}} = 0.8$).
Higher swap rates for the PT algorithm, marked by circles and by plus signs
in Figure $5$, result in a large decrease in the estimated inclusion probabilities for
the predictors corresponding to large regression coefficients $\beta_{\gamma}$.
The estimated inclusion probabilities for the predictors associated to low values of the regression coefficients
are the lowest for the MH algorithm whereas the PHS estimates are the highest
for these predictors. Finally, the plots on the right-hand side of Figure $5$ suggest that for this example the precision
of the three samplers, as measured by their MCSEs, is roughly comparable.
\begin{center}
  \begin{figure}[htbp]
   \epsfig{file=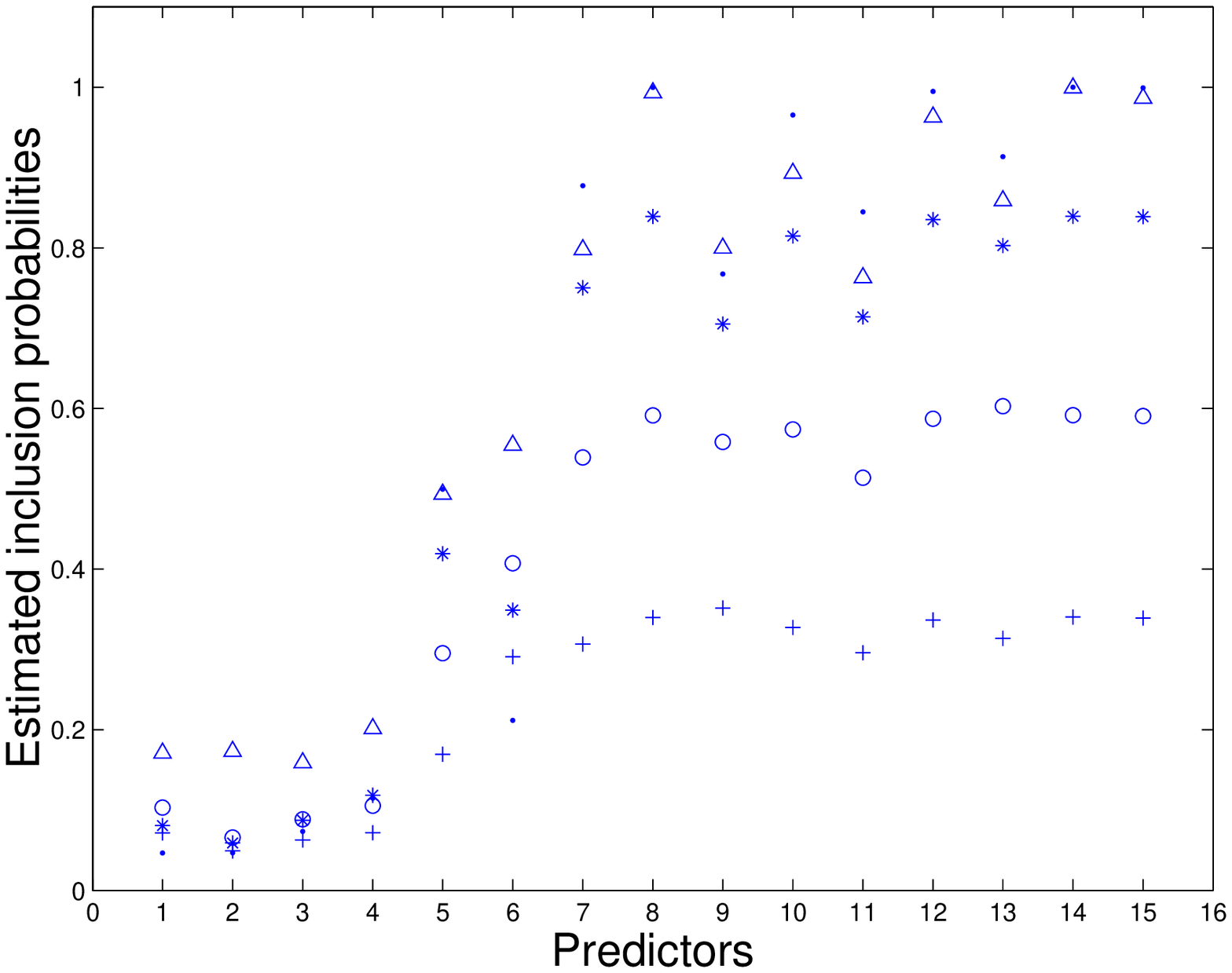,width=2.2in,height = 2.2in,angle=0}
   \epsfig{file=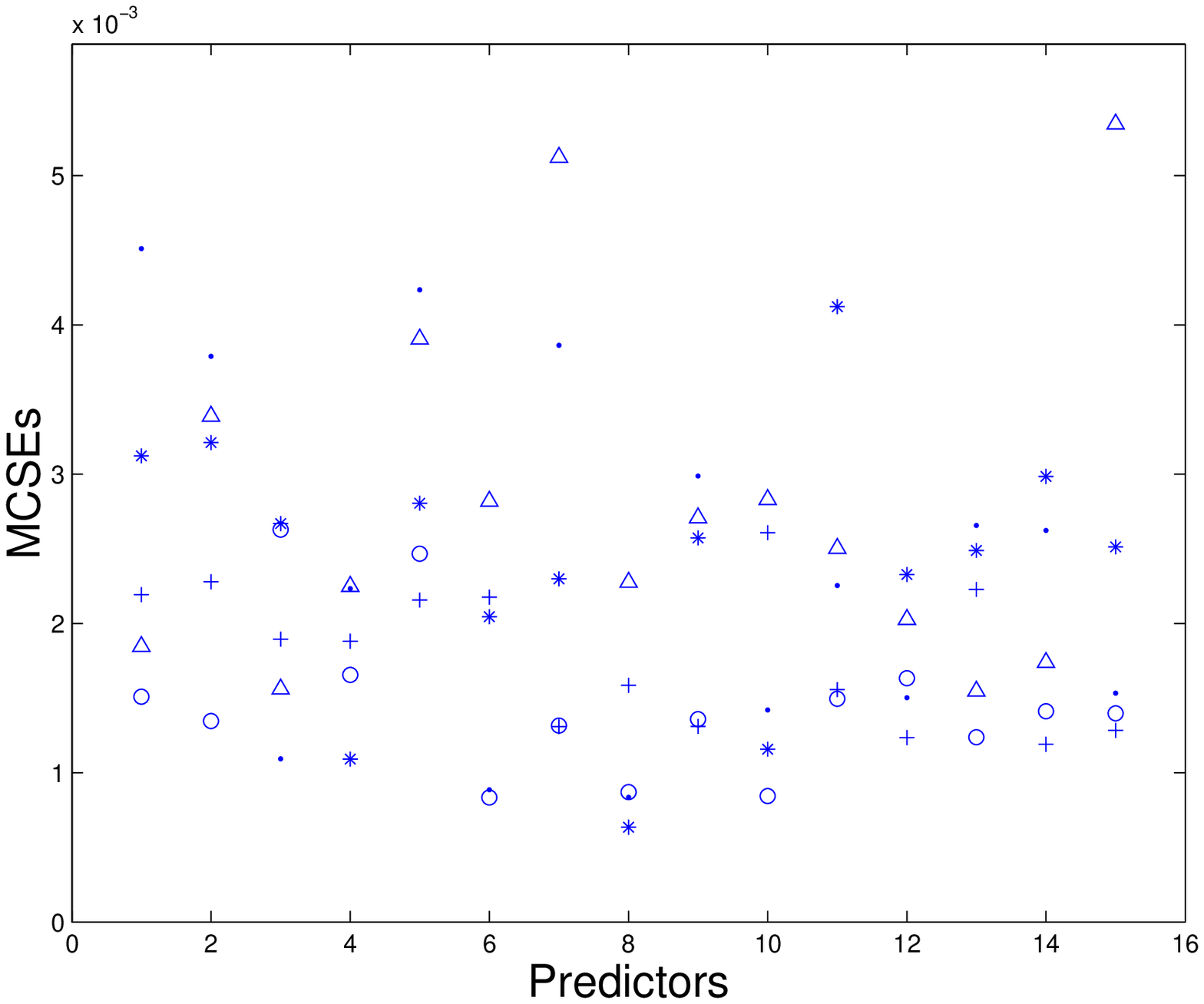,width=2.2in,height = 2.2in,angle=0}
   \epsfig{file=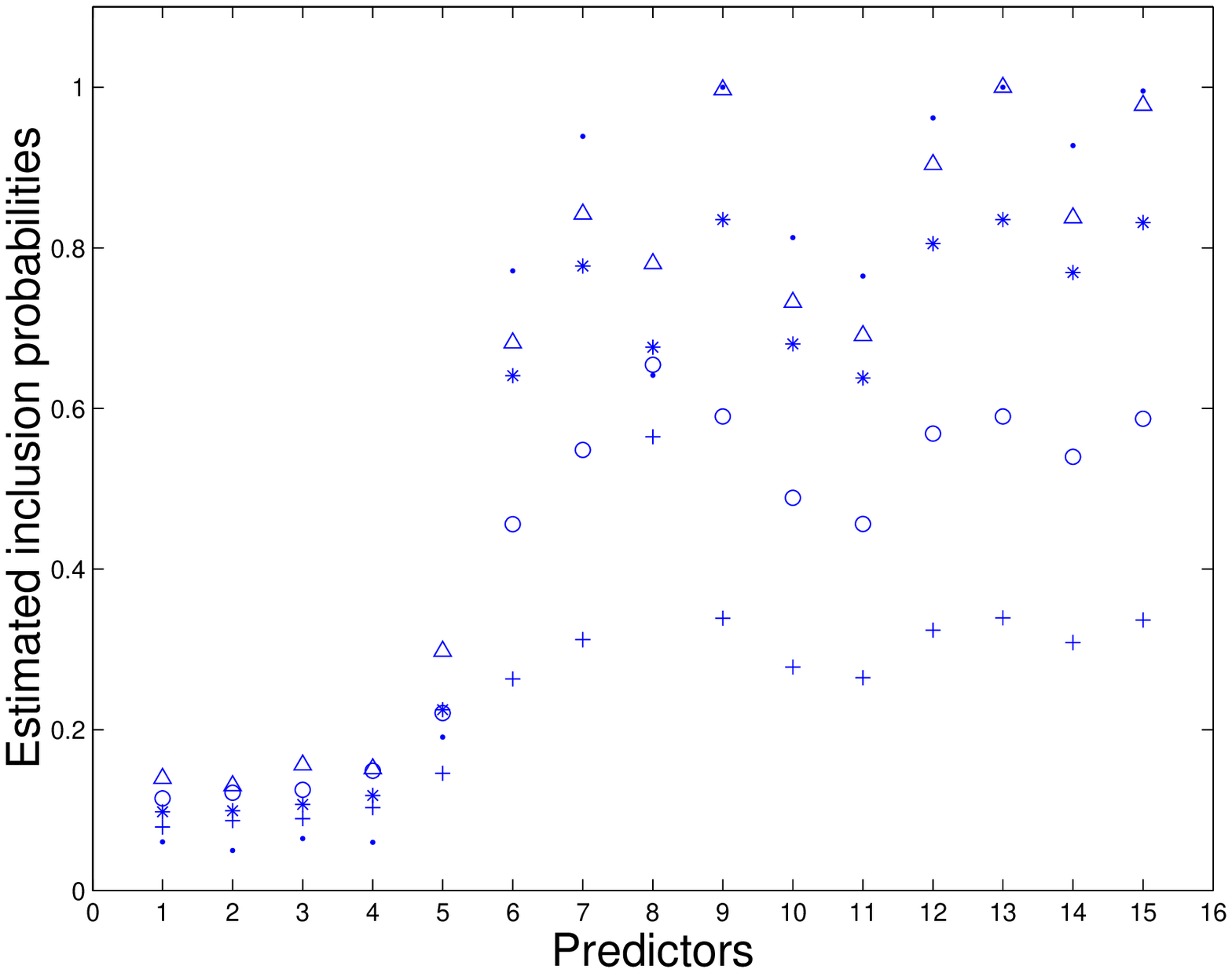,width=2.2in,height = 2.2in,angle=0}
   \hspace{0.5in}
   \epsfig{file=mcses_ncoll.eps,width=2.2in,height = 2.2in,angle=0}
    \caption{estimated marginal posterior inclusion probabilities for the fifteen predictors (left) and their Markov chain
     standard errors (right). The plots on top refer to the simulated data without collinearity and the bottom plots refer to the data
     with collinearity. In all plots, plus signs correspond to the results of the PT algorithm with $s = 0.8$, circles correspond to PT
    with $s = 0.5$ and asterisks mark the PT results with $s = 0.2$. Triangles identify the results of the PHS algorithm and dots mark
    those of MH. The PHS estimates for the marginal inclusion probabilities are the highest for the predictors associated to low values of their
     regression coefficients $\beta_{\gamma}$, whereas high swap rates for the PT algorithm produce wrong estimates of the inclusion probabilities
     for the predictors associated to high values of $\beta_{\gamma}$. The precision of the three algorithms, as measured by their MCSEs,
     appears to be comparable.}
  \end{figure}
\end{center}
\section{Application to the estimation of the structure of a survival CART model}
In regression and classification  trees (CART) the sample is clustered
in disjoint sets  called leaves. The leaves are  the final nodes
of a single-rooted  binary partition of the covariates  space which
we will refer  to as  the tree structure.   Within each  leaf, the
response  variable is  modeled according to  the regression,  or
classification or with the survival analysis frameworks (\cite{Breim}).
Bayesian  CART models  appeared in  the literature  with the papers of
\cite{chipman98bayesian} and \cite{Denison98}.   The MCMC model search
algorithms developed  in these two papers treat  the tree structure
as   an  unknown   parameter  and   explore  its   marginal  posterior
distribution using the Gibbs sampler and the MH algorithm.  In this example we focus on tree
models for randomly right-censored survival data (\cite{GordonOlshen},
\cite{DavisAnderson}, \cite{leblanch1},  \cite{leblanch2}).  The first
Bayesian survival tree model  has been proposed by \cite{MikeJen}, who
adopted a Weibull  leaf sampling density and a  step-wise greedy model
search algorithm  based on the  evaluation of all the  possible splits
within each node. The main strength of this model is that it incorporates
a flexible parametric form for the survival function and that the
tree search quickly converges to a mode in the model space.  In this Section, we propose  a fully Bayesian
analysis of the marginal posterior distribution over the space of tree
structures using PHS under the Weibull leaf likelihood.
Upon convergence, besides providing the structure of the estimated modal tree, the realisations of
the  cold   chain  provide  a  sample  from   the  marginal  posterior
distribution of the tree structure which can be employed to rank
the combinations  of the  covariates defining the  visited trees  as a
function  of their  estimated   marginal    posterior   inclusion
probabilities.
\subsection{Tree structure marginal posterior distribution}
Let  the survival times  $\{t_{j}\}_{j=1}^{n}$  be independent random variables
conditionally on  the tree  structure $(b,\zeta)$ and on  the Weibull leaf  parameters $(\alpha_{\zeta},\beta_{\zeta})$.
Under  this assumption,  the  joint sampling  density  of the
survival times can be written as
\begin{eqnarray}
f(t \mid X,\delta,b,\zeta,\alpha_{\zeta},\beta_{\zeta}) = \prod_{k = 1}^{b}\prod_{j = 1}^{n}\left(\left(\alpha_{k}\beta_{k}t_{j}^{\alpha_{k}-1}\right)^{\delta_{j}}e^{-\beta_{k}t_{j}^{\alpha_{k}}}\right)^{1_{k,j}},
\end{eqnarray}
where $b$ is the number of leaves, $\delta_{j}$ takes value $1$  for exact observations and $0$ for
right   censored   observations   and   $1_{k,j}  =   1_{\{X_{j}   \in
\zeta_{k}\}}$  is  $1$ if the
covariate  profile of  the $j$th  sample unit is included in $\zeta_{k}$, which is
the subset of the covariate space corresponding to  leaf $k = 1,...,b$, and $0$ otherwise.
Under a discrete uniform prior for the tree structure, the marginal
posterior probability  $P(b,\zeta \mid t, X,\delta)$  can be obtained,
up  to  a  multiplicative  constant,  by  integrating
$(11)$ with respect to the conditional prior distribution for the array of
leaf  parameters $(\alpha_{\zeta},\beta_{\zeta})$.   In  this work  we
place independent uninformative priors for each Weibull leaf parameter.
\cite{Sun}  provides an accurate
study of different uninformative  priors for the Weibull distribution.
In particular, Sun's paper indicates that $h_{k}(\alpha_{k},\beta_{k}) = 1/\alpha_{k}
\beta_{k}$ is the  Jeffreys' prior and the first  order matching prior
for the Weibull parameters of leaf $k$ given the parametrization $(11)$.  Sun  also shows
that, under $h_{k}(\alpha_{k},\beta_{k})$, the joint posterior density for
the leaf parameters $(\alpha_{k},\beta_{k})$
is proper when the number of data points is larger than  one and if all the observations falling
in leaf $k$ are  not equal. For this specification of the prior structure, the joint posterior of
the tree structure and of the leaf parameters can be written as
\begin{eqnarray}
f(b,\zeta,\alpha_{\zeta},\beta_{\zeta} \mid t,X,\delta) \propto \prod_{k = 1}^{b}\frac{1}{\alpha_{k} \beta_{k}}\prod_{j = 1}^{n}\left(\left(\alpha_{k}\beta_{k}t_{j}^{\alpha_{k}-1}\right)^{\delta_{j}}e^{-\beta_{k}t_{j}^{\alpha_{k}}}\right)^{1_{k,j}}.
\end{eqnarray}

The  Weibull   scale  parameters  $\beta_{k}$  can  be
integrated  out of  the joint  posterior $(12)$ analytically.   The resulting
integrated posterior density is
\begin{eqnarray}
f(b,\zeta,\alpha_{\zeta} \mid t, X,\delta) \propto \prod_{k}
\frac{\Gamma(\sum_{j}\delta_{j}1_{k,j})\alpha_{k}^{\sum_{j}\delta_{j}1_{k,j}-1}e^{(\alpha_{k}-1)\sum_{j}\delta_{j}\log(t_{j})1_{k,j}}}{(\sum_{j}t_{j}^{\alpha_{k}}1_{k,j})^{\sum_{j}\delta_{j}1_{k,j}}}.
\end{eqnarray}

Under $(13)$,  analytical integration of the  Weibull index parameters
$\alpha_{k}$ is not  possible.  For any given model, the Monte
Carlo   method  of   \cite{Chib}   can  be   employed   to compute   a
simulation-based approximation of  the marginal posterior probability
of the tree structure.
However, since this integration needs  to be performed for each visited
tree, the computational cost of Chib and Jeliazkov's method makes it unsuitable
for any iterative  model search.  In  this work  we
approximate  the tree structure  marginal posterior  probability using
the  Laplace  expansion  of  equation  $(13)$, which can be written as
\begin{eqnarray}
P(b,\zeta \mid t, X,\delta) \approx exp \left(b\frac{ \log(2\pi)}{2} + \sum_{k=1}^{b}\left( \log(l(\hat{\eta}_{k})) -\frac{\log(-l_{2}(\eta_{k}))\vert_{\hat{\eta}_{k}}}{2}\right)\right),
\end{eqnarray}
where $\hat{\eta}_{k}$ is the posterior  mode of the Weibull leaf  log index parameter
$\eta_{k} =  \log(\alpha_{k})$.
The derivation of equation $(14)$ and the explicit forms of the functions
$l(\eta)$ and $l_{2}(\eta)$ are reported in the Appendix.
\subsection{Marginal posterior inference for the tree structure}
In the  CART framework, as in the clustering problem illustrated in Section $3$,
the main challenge for constructing
efficient within-chain  proposal distributions is  the  lack  of a
distance  metric between  different models.   This issue has been also noted  by
\cite{Brooks}  in the context of the reversible  jump  MCMC  algorithm
(\cite{Green}).  Our   specification  of  the within-chain  proposal  distribution
generalizes     the     approaches     of     \cite{Denison98}     and
\cite{chipman98bayesian}    by   devising  two  additional   within-chain
transitions besides their \emph{insert}, \emph{delete} and \emph{change} moves. For the  within-chain updates
we propose a  transition  at random  among  the following five types:
\begin{itemize}
\item[1)] Insert: sample a leaf at random and insert a
  new split by randomly selecting a new splitting rule.
\item[2)] Delete: sample at random a leaf pair with common parent and at most one child split and delete it.
\item[3)] Change: resample at random one splitting rule.
\item[4)] Permute: sample a random number of splits
  and permute at random  their splitting rules.
\item[5)] Graft: sample at random one
  of the tree branches and graft it to one of the leaves of a different branch.
\end{itemize}
\cite{chipman98bayesian} noted that their  MCMC algorithm can effectively  resample the splitting
rules of nodes close to the  tree leaves but the rules defining splits
close to  the tree root are  seldom replaced. In  our specification of
the  within-chain  transitions, move  number $4$ aims at
improving  sampling of the  splitting rules  at all  levels of  the tree
structure.  Furthermore,  the fifth  move  type allows  the sampler  to
jump  to a  tree structure distinct  from the  current one
without changing its splitting rules.

Having  adopted a multiple-chains algorithm,  we also devised
two types of cross-chains  transitions. The first  is the
cross-chains version of the insert, graft and change transitions, swapping
the elements of the tree structure required to perform
corresponding pairs of transitions across chains. The second
class of cross-chains transitions includes a whole
tree swap  between chains. 

At iteration $i$, the PHS algorithm for this example proceeds as follows:
\begin{itemize}
\item[1)] choose at random one of the heated chains $m_{i} \in [2,M]$
  and propose at random
      one of the cross-chains moves, accepting the swap with
  probability $1$.
\item[2)] update each of the remaining $M-2$ chains independently using the
five types of within-chain transitions and the MH acceptance probability.
\end{itemize}
\subsection{Analysis of a set of cancer survival times}
Colorectal  adenocarcinoma ranks  second as  a cause  of death  due to
cancer in the western world and  liver metastasis is the main cause of
death in  patients with colorectal  cancer (\cite{PasettoMonfardini}).
The survival times of $622$ patients
with liver  metastases from a colorectal primary  tumor were collected
along with  their clinical  profiles by the  International Association
Against   Cancer  (\texttt{http://www.uicc.org}).    Table  $1$   reports   a
description of  the nine available clinical covariates.  The survival  times of this
dataset  are currently included  in the  \textbf{R} library  \emph{locfit}
(\texttt{http://www.locfit.info}). This
data  has  been   analyzed  by  \cite{Hermanek}  using  non-parametric
methods,  by \cite{Antoniadis}  using their  wavelet-based  method for
estimating the survival density  and the instantaneous hazard function
and by \cite{KottasMixWeib}, who  employed a Dirichlet process mixture
of Weibull distributions to derive a Bayesian non-parametric estimate
of   the    survival   density    and   of   the    hazard   function.
\cite{HauptMansmann}   employed  this   dataset   to  illustrate   the
non-parametric tree fitting techniques for survival data implemented in
the \texttt{S-plus} function \emph{survcart}.  The aim of this Section
is showing  that the estimates  of $(b,\zeta)$ obtained using  the PHS
algorithm  and the approximate  marginal posterior $(14)$
provide meaningful inferences for the prognostic significance  of the
available  covariates.   For  each   covariate,  the  latter  will  be
represented by its estimated  posterior inclusion probability, i.e. by
the proportion of sampled models which structure depends on the covariate.
\begin{table}[htbp]
\begin{center}
  \begin{tabular}{|c|c|c|c|}    \hline
&\textbf{Name}& \textbf{Symbol} & \textbf{Description} \\\hline
$1$ & Diam. largest LM & DLM& $(1,20)$mm\\\hline
$2$ & Age &  AGE & $(18,88)$years\\\hline
$3$ & Diagnosis of LM & TD & synchrone/metachron with CPT\\\hline
$4$ & Gender & SEX & M = $55.8\%$, F = $44.2\%$\\\hline
$5$ & Lobar involvement& LI & unilobar/bilobar\\\hline
$6$ & Number of LM& NLM & $(1,20+)$\\\hline
$7$ & Locoregional disease& LRD &yes/no\\\hline
$8$ & Metastatic stage & TNM & local/regional/distant\\\hline
$9$ & Location PT& LOC &colon/rectum\\\hline
\end{tabular}\\
\caption{description of the covariates for the liver dataset.
The data include several types of clinical covariates, such as continuous (DLM),
discrete (AGE, NLM) and categorical (all others).}
\end{center}
\end{table}
A  PHS using twenty  parallel  chains  was  run  for  fifty  thousand
iterations, the  starting tree  for each chain  being the  root model.
Consistently with  the PHS algorithm  described in Section  $5.2$, for
this analysis we used a uniform swap proposal distribution $q_{s}(\cdot)$.
On the top row, Figure $6$  shows the unnormalised log posterior  tree probability for
the  models visited by the cold chain,  plotted respectively  versus the
iteration index and versus their number of leaves.  The  posterior sampling  for the  cold chain
moved quickly towards  areas of  high marginal  posterior probability  models, which
leaf range is  $10-14$, the best tree having  $12$ leaves. The bottom plot of Figure $6$
shows the estimated  marginal  posterior inclusion  probabilities for all the
covariates. According to these estimates, the covariates with maximal
prognostic significance are  the  diameter  of   the  largest  liver  metastasis
and the number   of  liver metastases, followed by their locoregional disease status, the patients'age
at diagnosis, the localisation  of
their  primary tumor and their lobar  involvement  status. The estimated inclusion probabilities
of the remaining covariates suggest that, for this sample, their values do
not discriminate among significantly different survival clusters.
\begin{center}
  \begin{figure}[htbp]
    \epsfig{file=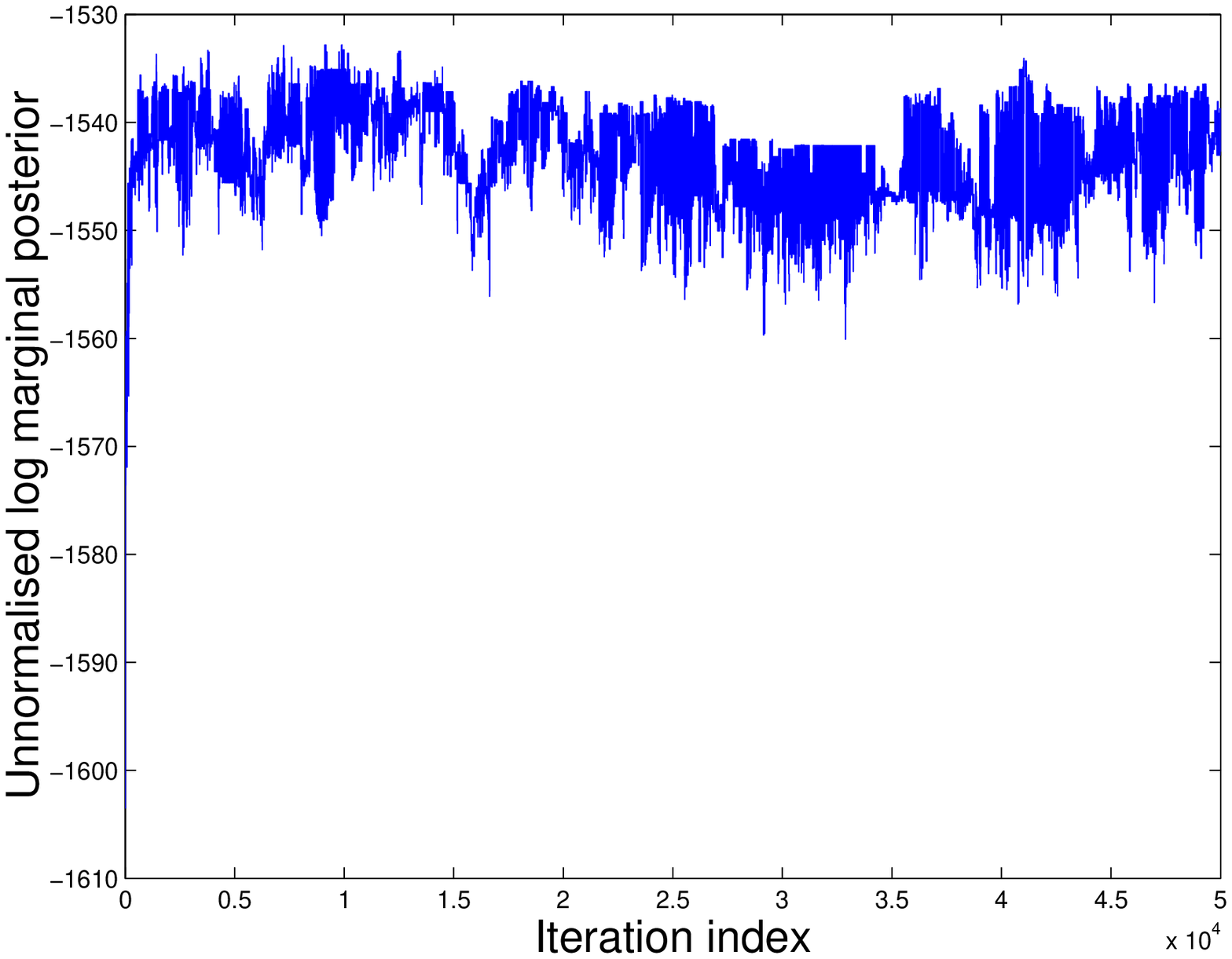,width=2.5in,height = 2.5in,angle=0}
    \epsfig{file=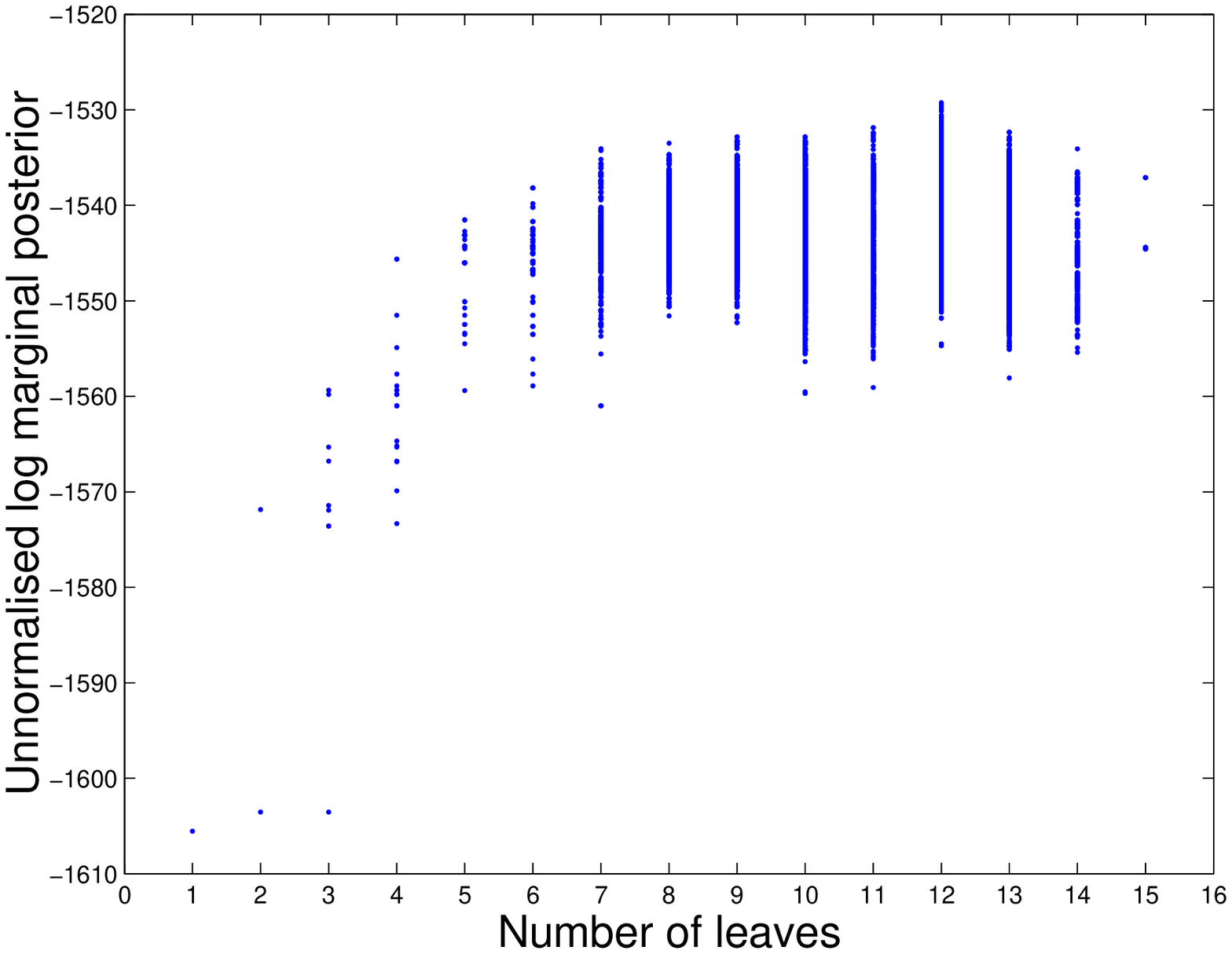,width=2.5in,height = 2.5in,angle=0}
    \epsfig{file=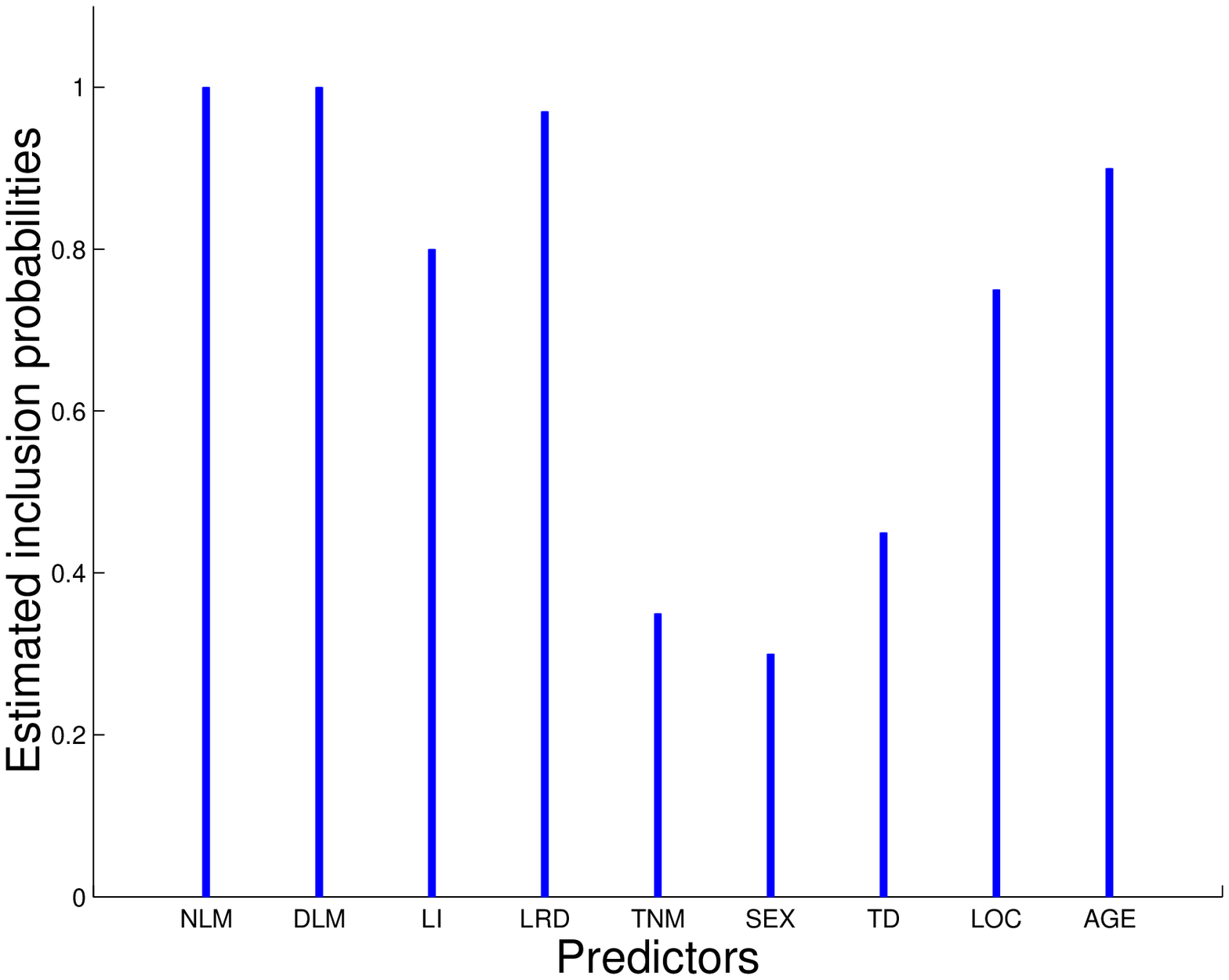,width=2.5in,height = 2.5in,angle=0}
    \caption{the top plots show the unnormalised log marginal posterior probability of the tree structure for the models
     visited along the PHS posterior simulation by the cold chain. The horizontal axis in the left plot
     represents the iteration index, whereas in the right plot it represents
     the number of leaves of the corresponding tree. The plot on the bottom
     shows the estimated marginal posterior inclusion probabilities for the nine 
     covariates. Those with maximal prognostic significance are  the  diameter  of   the  largest  liver  metastasis,
     the number of  liver metastases, the locoregional disease status, the patients'age
     at diagnosis, the localisation  of their  primary tumor and their lobar  involvement  status.}
  \end{figure}
\end{center}
Figure   $7$  shows  the structure of the estimated  modal
posterior tree.  The depth of the  leaves in this  figure reflects the
number  of splits  required to  generate them.
For any given tree structure $(b,\zeta)$, posterior inferences
for $(\alpha_{\zeta},\beta_{\zeta})$ can be obtained by further simulation
using their full conditional posterior distributions, which can be
easily derived from the full conditional posterior $(12)$. In order to maintain the
focus of this Section on the application of PHS for tree selection, here
we adopt the Kaplan-Meier (KM) survival curves as non-parametric estimates
of the survival probabilities within each leaf.
Table $2$ reports the number of patients clustered within
each leaf of the modal tree and the values  of their KM survival probabilities at $12$,
$24$  and $36$  months.   The  bold figures  in  Table $2$  correspond
respectively  to  the highest and  to  the lowest estimated  survival
probabilities at the three time points. The lowest survival correspond
to leaves number  $1$ and $2$ for all the three time points.
These two groups are defined by high values of the first two
covariates  in   the  tree, which are the diameter of the largest liver metastasis and
the number of liver metastases.  The highest  estimated  survival
probabilities at $12$, $24$ and $36$ months correspond to  leaf number  $8$  which is
characterised by at most one liver metastasis of small diameter, local
spreading of the cancer and no locoregional disease.
\begin{table}[htbp]
\begin{center}
  \begin{tabular}{|c|c|c|c|c|}    \hline
 Leaf Number & Cluster size & $\hat{P}_{km}(t \geq 12)$ & $\hat{P}_{km}(t \geq 24)$ & $\hat{P}_{km}(t \geq 36)$\\\hline
$1$ & $ 63$ & $ 0.70 $ & $0.23$ & $ \mathbf{0.03}$\\\hline
$2$ & $ 148 $ & $\mathbf{0.58} $ & $\mathbf{0.14} $ & $ \mathbf{0.04}$\\\hline
$3$ & $ 34 $ & $ 0.75 $ & $0.53$ & $ 0.43$\\\hline
$4$ & $ 78 $ & $ 0.85 $ & $0.50 $ & $0.28 $\\\hline
$5$ & $ 42 $ & $ 0.85$ & $ 0.50 $ & $0.16$\\\hline
$6$ & $ 48 $ & $ 0.80$ & $ 0.57$ & $ 0.26$\\\hline
$7$ & $ 31 $ & $ 0.90$ & $ 0.81$ & $ 0.64$\\\hline
$8$ & $ 42 $ & $ \mathbf{1.00} $ & $\mathbf{0.84}$ & $ \mathbf{0.77}$\\\hline
$9$ & $ 30$ & $  0.69$ & $ 0.25$ & $ 0.19$\\\hline
$10$ & $ 32$ & $ 0.65$ & $ 0.30$ & $ 0.10$\\\hline
$11$ & $ 42$ & $ 0.68 $ & $0.59$ & $ 0.29$\\\hline
$12$ & $ 31$ & $ 0.97$ & $ 0.76$ & $ 0.29$\\\hline
\end{tabular}\\
\caption{number of observations falling in each leaf
         of the estimated posterior modal tree and Kaplan-Meier
         survival probabilities at $12$, $24$ and $36$ months. The highest
         survival probabilities correspond to leaf number $8$, which is
         defined by a low number of local metastases of small diameter and
         by the absence of locoregional disease. The lowest survival probabilities
         correspond to leaves $1$ and $2$, which are characterized respectively by
         metastases of large diameter and by a large number of smaller metastases.}
\label{tab:srules}
\end{center}
\end{table}
\begin{center}
  \begin{figure}[htbp]
   \epsfig{file=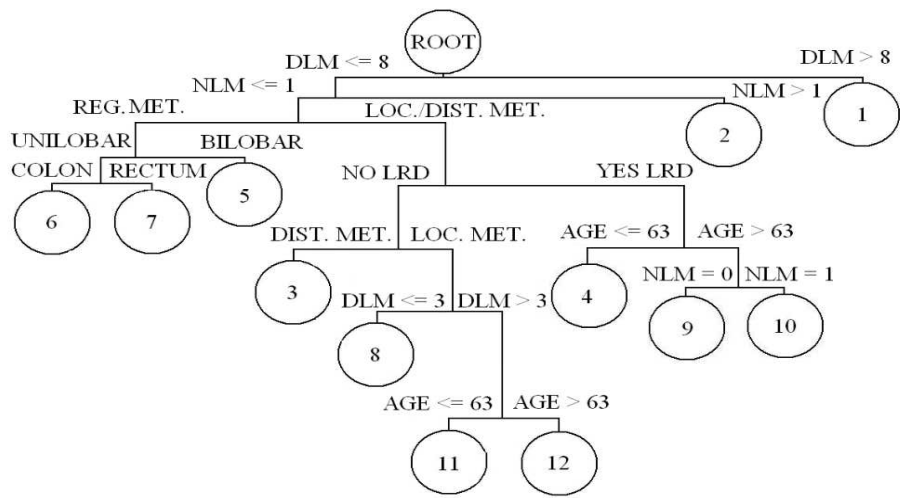,width=2.6in,height = 5in,angle=-90}
    \caption{tree structure of the estimated modal tree found by the PHS simulation. The lowest estimated
    survival probabilities at $1$, $2$ and $3$ years correspond to leaves number $1$ and $2$, which are
    defined by high values of the two key covariates (NLM,DLM), whereas the
    best estimated survival corresponds to leaf number $8$, which is defined by the non-linear
    interaction of (DLM,NLM,LRD and LOC).}
  \end{figure}
\end{center}
\section{Discussion}
This paper presents a novel multiple chains algorithm for Markov chain Monte Carlo
inference, which we labeled parallel hierarchical sampler.
We emphasized the main terms of comparison to evaluate the parallel hierarchical sampler, that
are the Metropolis algorithm, the muliple-try Metropolis algorithm of
\cite{mtry} and parallel tempering. Whilst 
in the Metropolis sampler high acceptance rates usually produce
high autocorrelations and slow convergence (\cite{Roberts97}), by
construction PHS produces a chain which always moves but which
exhibits low serial dependence. As it can be seen be comparing the PT
and PHS joint transition kernels, reported by equations $(6)$ and $(8)$,
the first advantage of PHS with respect to PT is that its swap
proposal has a simpler form. This is because in PHS both types
of transitions are performed at each iterations instead of
being sampled according to the distribution $q_{s}^{'}(s_{i} \mid s_{i-1})$.
The second practical advantage of PHS with respect to PT
is that its implementation does not require choosing a set of
temperature values. 
When little is known about the equilibrium
distribution being studied, we regard this feature of PHS as a
potentially major advantage.
Furthermore, although the algorithm presented
in Section $2$ allows for a different proposal distribution within
each of the chains indexed $m = 2,...,M$, this is not even a necessary
requirements for the implementation of PHS.

As pointed out by \cite{Geyer},
the attractive feature of multiple-chains MCMC samplers such as PT and PHS
is that their target distribution factors into the product of the marginal
distributions for each chain despite the fact that these chains are made
dependent by the swap transitions. Under the conditions
of Section $2$ we prove in Theorem $1$ that the samples generated
by the PHS algorithm converge weakly to such product distribution.

In Section $2$ we also noted that the complexity of multiple-chains transition kernels,
which for PT and PHS are mixtures of their marginal
transition kernels, largely prevents a direct analytical comparison
of their convergence properties. Direct comparison of the transition
kernels $(6)$ and $(8)$ leads to major analytical difficulties
and so far it has not been possible to establish an ordering
between the two kernels using the criteria illustrated in \cite{Peskun}, \cite{MT} and
\cite{Mira2}.  Although it falls
beyond the scope of this paper, we regard
the development of computable ordering criteria for mixtures of Markov chain
transition kernels as a key area which deserves further investigation.

Following \cite{Huel}, the last three Sections of this paper emphasise the relevance of
multiple-chains MCMC algorithms for estimating the posterior model probabilities in a
variety of settings. In Sections $3$ and $4$ we provide
two examples comparing numerically the inferences obtained by the
Metropolis-Hastings algorithm with those of PT and PHS.
In these two Sections we focussed on comparing
the posterior inferences without considering the computational time
required to produce them. The rationale behind this choice is that
the computational time required by multiple-chains samplers is highly
dependent on the available computational resources. For instance, if
all chains used by PT and PHS are run in parallel on several
processors, their run time may be comparable to that of the single-chain
Metropolis-Hastings sampler, whereas if all chains are updated sequentially
using one processor their run time is, of course, much longer.

In Section $3$ we observed that for the Gaussian clustering
the PHS algorithm appears to explore
the space of cluster configurations more effectively with respect to
MH and PT using the same within-chain proposal mechanism for all
samplers.
For the example of Section $4$ we observed that using high swap proposal rates for the PT sampler
leads to wrong estimates of the marginal posterior inclusion probabilities
with and without collinearity among the predictors $X$. However, by
measuring the precision of the three MCMC algorithms by their
Markov chain standard errors we did not find a significant advantage
of the multiple-chains samplers with respect to the MH algorithm.

Section $5$ illustrates the application of PHS for deriving
inferences for the structure of a treed survival model.
One of the main differences between
of Sections $4$ and  $5$ is that
the focus of the former is the selection of the relevant main
regression effects
whereas in the latter the key elements defining different survival
groups are the non-linear interactions among the predictors defining the tree structure.
The top-right plot in Figure $6$ shows that the approximated marginal posterior probability of the tree structure
under the Weibull model does not increase monotonically with the number of leaves.
Under the Weibull model we used the Laplace expansion to approximate the tree
structure marginal posterior probability. The Schwarz approximation
(\cite{Schwarz}) was also considered. The penalty term of the  Schwarz approximation
increases  with the  model dimension,  thus it  represents a  cost for
complexity  factor.   However,  given  a  fixed  number of leaves  this
approximation favours trees allocating the data more
unevenly   across   leaves.     Therefore,   employing   the   Schwarz
approximation when many covariates are available might  result in assigning
significant  posterior  probability  to  large and  unbalanced  trees,
leading  to overfitting small  groups of  survival data. On the other hand, the penalty
associated to  the Laplace approximation has a  complex form involving
the tree size, the log  Weibull parameters $\eta_{i}$ and the survival
times  along  with their  censoring  indicators.   Evaluation of  this
approximation  for a  variety  of tree  structures  showed that  this
penalty is strictly increasing with the tree dimension but it does not
favour   unbalanced   trees.   Using the Laplace expansion to
approximate the model's marginal posterior probability
and the PHS algorithm to sample from it, we find meaningful
posterior inferences for a set of colorectal cancer survival data.
The estimated modal tree separates the short-term
survivors, who are characterised by a large number of liver metastases
of large size, from the long-term survivors, who present a few local
metastases of small size without further symptoms.

Finally, in Sections $3$, $4$ and $5$ we addressed qualitatively the issue of
convergence of the chains produced by the three MCMC algorithms
by considering their acceptance rates
and the fluctuations of their marginal posterior probabilities. Being the model spaces inherently non-metric, it was not possible to
use the state-dependent criteria commonly used  to assess
the convergence of Markov chains to their stationary distributions
illustrated in \cite{Carlin}, \cite{RGG}
or in \cite{RobDiag} among others. Although it
is beyond the scope of this paper, in light of the increasing relevance
of model selection problems we consider the development of
appropriate convergence measures an important field for future research.
\section*{Appendix}
The Laplace approximation is the  second order Taylor expansion of the
logarithm of the integrated  posterior $(14)$ around its posterior  mode.  In order
to  derive  the  approximation,  it  is  convenient  to  parametrize
equation $(14)$ as a function of $\eta_{k} = \log(\alpha_{k})$, so that
the    variables     to    be    integrated     out    have    support
 on the real line.  Under this  parametrization, stable  estimates of
the  posterior  modes  $\{\hat{\eta}_{k}\}_{k=1}^{b}$ can  be  computed
numerically.   For   each  leaf  the  log  integrated
conditional posterior is
\begin{eqnarray*}
l(\eta_{k}) &\propto& \log(\Gamma(\sum_{j}\delta_{j}1_{k,j})) + (\eta_{k}-1) \sum_{j}\delta_{j}1_{k,j} + e^{\eta_{k}}\sum_{j}\delta_{j}\log(t_{j})1_{k,j} \\
&-& \sum_{j}\delta_{j}1_{k,j}\log(\sum_{j}t_{j}^{e^{\eta_{k}}}1_{k,j}).\nonumber
\end{eqnarray*}

The penalty arising from  the Laplace approximation is proportional to
minus the  logarithm of  the second derivative  of the  log integrated
posterior taken  with respect  to the leaf  parameters $\{\eta_{k}\}$.
The second derivative of the function $l(\eta_{k})$ is
\begin{eqnarray*}
l_{2}(\eta_{k}) =  e^{\eta_{k}} \left(\sum_{j}\delta_{j}\log(t_{j})1_{k,j}
- \frac{\sum_{j}\delta_{j}1_{k,j} \sum_{j}t_{j}^{e^{\eta_{k}}}(\log(t_{j}))^{2}1_{k,j}}{\sum_{j}t_{j}^{e^{\eta_{k}}}1_{k,j}}\right).
\end{eqnarray*}
Summing the approximation over the $b$ leaves yields the right-hand side of equation $(15)$.
\small
\bibliography{Rigat_phs}

\begin{thebibliography}{58}
\providecommand{\natexlab}[1]{#1}
\providecommand{\url}[1]{\texttt{#1}}
\expandafter\ifx\csname urlstyle\endcsname\relax
  \providecommand{\doi}[1]{doi: #1}\else
  \providecommand{\doi}{doi: \begingroup \urlstyle{rm}\Url}\fi

\bibitem[Antoniadis et~al.(1999)Antoniadis, Gr\'egoire, and Nason]{Antoniadis}
A.~Antoniadis, G.~Gr\'egoire, and G.~Nason.
\newblock Density and hazard rate estimation for right-censored data by using
  wavelet methods.
\newblock \emph{Journal of the {R}oyal {S}tatistical {S}ociety, series {B}},
  61:\penalty0 63--84, 1999.

\bibitem[Billera and Diaconics(2001)]{BilleraDiac}
L.J. Billera and P.~Diaconics.
\newblock A geometric interpretation of the {M}etropolis-{H}astings algorithm.
\newblock \emph{Statistical Science}, 16:\penalty0 335--339, 2001.

\bibitem[Breiman et~al.(1984)Breiman, Friedman, Olshen, and Stone]{Breim}
L.~Breiman, J.H. Friedman, R.A. Olshen, and C.J. Stone.
\newblock \emph{Classification and Regression Trees}.
\newblock Chapman and Hall, New York, 1984.

\bibitem[Brooks et~al.(2003)Brooks, Giudici, and Roberts]{Brooks}
S.P. Brooks, P.~Giudici, and G.O. Roberts.
\newblock Efficient construction of {R}eversible {J}ump {MCMC} proposal
  distributions.
\newblock \emph{Journal of the {R}oyal {S}tatistical {S}ociety, series {B}},
  65:\penalty0 3--55, 2003.

\bibitem[Carlin and Chib(1995)]{CarlinChib95}
B.P. Carlin and S.~Chib.
\newblock Bayesian model choice via {M}arkov chain {M}onte {C}arlo methods.
\newblock \emph{Journal of the {R}oyal {S}tatistical {Society}, series {B}},
  57, 1995.

\bibitem[Carlin and Cowles(1996)]{Carlin}
B.P. Carlin and M.K. Cowles.
\newblock {M}arkov chain {M}onte {C}arlo convergence diagnostics: A comparative
  review.
\newblock \emph{Journal of the {A}merican {S}tatistical {A}ssociation},
  91:\penalty0 883--904, 1996.

\bibitem[Chib and Greenberg(1995)]{Greenberg}
S.~Chib and E.~Greenberg.
\newblock Understanding the {M}etropolis-{H}astings algorithm.
\newblock \emph{The {A}merican {S}tatistician}, 49:\penalty0 327--335, 1995.

\bibitem[Chib and Jeliazkov(2001)]{Chib}
S.~Chib and I.~Jeliazkov.
\newblock Marginal likelihood from the {M}etropolis-{H}astings output.
\newblock \emph{Journal of the {A}merican {S}tatistical {S}ociety},
  96:\penalty0 270--281, 2001.

\bibitem[Chipman et~al.(1998)Chipman, George, and McCulloch]{chipman98bayesian}
Hugh~A. Chipman, Edward~I. George, and Robert~E. McCulloch.
\newblock Bayesian {CART} model search.
\newblock \emph{Journal of the {A}merican {S}tatistical {S}ociety},
  93:\penalty0 935--947, 1998.

\bibitem[Clyde and George(2004)]{clydegeorge04}
M.~Clyde and E.I. George.
\newblock Model {U}ncertainty.
\newblock \emph{Statistical Science}, 19:\penalty0 81--94, 2004.

\bibitem[Davis and Anderson(1989)]{DavisAnderson}
R.B. Davis and J.R. Anderson.
\newblock Exponential survival trees.
\newblock \emph{Statistics in Medicine}, 8, 1989.

\bibitem[Dellaportas et~al.(2002)Dellaportas, Forster, and
  Ntzoufras]{DellapForster02}
P.~Dellaportas, J.J. Forster, and I.~Ntzoufras.
\newblock On {B}ayesian model and variable selection using {MCMC}.
\newblock \emph{Statistics and {C}omputing}, 12:\penalty0 27--36, 2002.

\bibitem[Denison et~al.(1998)Denison, Mallick, and Smith]{Denison98}
D.G.T. Denison, B.K. Mallick, and A.F.M. Smith.
\newblock A {B}ayesian {CART} algorithm.
\newblock \emph{Biometrika}, 85, 1998.

\bibitem[Gamerman(1997)]{Dani}
D.~Gamerman.
\newblock \emph{Markov chain Monte Carlo: Stochastic Simulation for Bayesian
  Inference}.
\newblock Chapman and Hall, 1997.

\bibitem[Gelfand and Smith(1990)]{GelfSmith}
A.E. Gelfand and A.~F.~M. Smith.
\newblock Sampling-based approaches to calculating marginal densities.
\newblock \emph{Journal of the {A}merican {S}tatistical {S}ociety},
  85:\penalty0 398--409, 1990.

\bibitem[George and McCulloch(1997)]{GeorgeMcCull97}
E.I. George and R.E. McCulloch.
\newblock Approaches for {B}ayesian variable selection.
\newblock \emph{Statistica Sinica}, 7:\penalty0 339--373, 1997.

\bibitem[George and McCulloch(1993)]{SSVS}
E.I. George and R.E. McCulloch.
\newblock Variable selection via {G}ibbs sampling.
\newblock \emph{Journal of the {A}merican {S}tatistical {A}ssociation},
  88:\penalty0 882--889, 1993.

\bibitem[Geweke(1992)]{Geweke}
J.~Geweke.
\newblock Evaluating the accuracy of sampling-based approaches to the
  calculation of posterior moments.
\newblock \emph{In: Bayesian Statistics 4, J.M. Bernardo, J.O. Berger, A.P.
  Dawid and {A}.{F}.{M}. {S}mith editors; {O}xford Press}, pages 169--194,
  1992.

\bibitem[Geyer and Thompson(1995)]{GThompson}
C.~J. Geyer and E.A. Thompson.
\newblock Annealing {M}arkov {C}hain {M}onte {C}arlo with applications to
  ancestral inference.
\newblock \emph{Journal of the {A}merican {S}tatistical {A}ssociation},
  90:\penalty0 909--920, 1995.

\bibitem[Geyer(1991)]{Geyer}
C.J. Geyer.
\newblock Markov chain {M}onte {C}arlo maximum likelihood.
\newblock \emph{Computing Science and Statistics: Proceedings on the 23rd
  Symposium on the Interface, New York}, 1991.

\bibitem[Gilks et~al.(1995)Gilks, Richardson, and Spiegelhalter]{Gilks}
W.R. Gilks, S.~Richardson, and D.~J. Spiegelhalter.
\newblock \emph{Markov chain Monte Carlo in practice}.
\newblock Chapman and Hall, 1995.

\bibitem[Gordon and Olshen(1995)]{GordonOlshen}
L.~Gordon and R.A. Olshen.
\newblock Tree-structured survival analysis.
\newblock \emph{Cancer Treatment Reports}, 69, 1995.

\bibitem[Green(1995)]{Green}
P.H. Green.
\newblock Reversible {J}ump {M}arkov chain {M}onte {C}arlo computation and
  {B}ayesian model determination.
\newblock \emph{Biometrika}, 82:\penalty0 711--732, 1995.

\bibitem[Hastings(1970)]{MH}
W.K. Hastings.
\newblock Monte {C}arlo sampling methods using {M}arkov chains and their
  applications.
\newblock \emph{Biometrika}, 57-1:\penalty0 97--109, 1970.

\bibitem[Haupt and Mansmann(1995)]{HauptMansmann}
G.~Haupt and U.~Mansmann.
\newblock Survival trees in {S}plus.
\newblock \emph{Advances in Statistical Software 5}, pages 615--622, 1995.

\bibitem[Hermanek and Gall(1990)]{Hermanek}
P.~Hermanek and F.P. Gall.
\newblock Uicc {S}tudie zur {K}lassifikation von {L}ebermetastasen.
\newblock \emph{Chirurgie der {L}ebermetastasen und primaren malignen
  {T}umoren}, 1990.

\bibitem[Huelsenbeck et~al.(2001)Huelsenbeck, Ronquist, Nielsen, and
  Bollback]{Huel}
J.P. Huelsenbeck, F.~Ronquist, R.~Nielsen, and J.P. Bollback.
\newblock Bayesian inference of {P}hylogeny and its impact on evolutionary
  biology.
\newblock \emph{Science}, 14:\penalty0 2310--2314, 2001.

\bibitem[Hukushima and Nemoto(1996)]{HN}
K.~Hukushima and K.~Nemoto.
\newblock Exchange {M}onte {C}arlo method and application to {S}pin {G}lass
  simulations.
\newblock \emph{Journal of the Physical Society of Japan}, 65:\penalty0
  1604--1620, 1996.

\bibitem[Kottas(2003)]{KottasMixWeib}
A.~Kottas.
\newblock Nonparametric {B}ayesian survival analysis using mixtures of
  {W}eibull distributions.
\newblock \emph{To appear in the {J}ournal of {S}tatistical {P}lanning and
  {I}nference}, 2003.

\bibitem[Kuo and Mallick(1998)]{KuoMallick98}
L.~Kuo and B.~Mallick.
\newblock Variable selection for regression models.
\newblock \emph{Sankhya, {B}}, 60:\penalty0 65--81, 1998.

\bibitem[Liu(2001)]{Liu}
J.S. Liu.
\newblock \emph{Monte Carlo Strategies in Scientific computing}.
\newblock Springer, 2001.

\bibitem[Liu et~al.(2000)Liu, Liang, and Wong]{mtry}
J.S. Liu, F.~Liang, and W.H. Wong.
\newblock The multiple-try method and local optimization in {M}etropolis
  sampling.
\newblock \emph{ournal of the {A}merican {S}tatistical {A}ssociation},
  95:\penalty0 121--134, 2000.

\bibitem[Metropolis and Ulam(1949)]{Ulam}
N.~Metropolis and S.~Ulam.
\newblock The {M}onte {C}arlo method.
\newblock \emph{Journal of the {A}merican {S}tatistical {A}ssociation},
  44:\penalty0 335--341, 1949.

\bibitem[Metropolis et~al.(1953)Metropolis, Rosenbluth, Rosenbluth, Teller, and
  Teller]{Metropolis}
N.~Metropolis, A.~Rosenbluth, M.~Rosenbluth, M.~Teller, and E.~Teller.
\newblock Equations of state calculations by fast computing machines.
\newblock \emph{J. Chem. Phys.}, 21:\penalty0 1087--2092, 1953.

\bibitem[Meyn and Tweedie(1994)]{MT}
S.~Meyn and R.L. Tweedie.
\newblock State-dependent criteria for convergence of {M}arkov chains.
\newblock \emph{The {A}nnals of {A}pplied {P}robability}, 4:\penalty0 149--168,
  1994.

\bibitem[Mira(2001)]{Mira2}
A.~Mira.
\newblock Ordering and improving the performance of {M}onte {C}arlo {M}arkov
  chains.
\newblock \emph{Statistical Science}, 16:\penalty0 340--350, 2001.

\bibitem[Mira and Geyer(1999)]{Mira3}
A.~Mira and C.~Geyer.
\newblock Ordering {M}onte {C}arlo {M}arkov chains.
\newblock \emph{Technical Report 632, School of Statistics, University of
  Minnesota}, 1999.

\bibitem[Mitchell and Beauchamp(1988)]{Mitchell}
T.J. Mitchell and J.J. Beauchamp.
\newblock Bayesian variable selection in linear regression.
\newblock \emph{Journal of the {A}merican {S}tatistical {A}ssociation},
  83:\penalty0 1023--1032, 1988.

\bibitem[M.Leblanch and J.Crowley(1992{\natexlab{a}})]{leblanch1}
M.Leblanch and J.Crowley.
\newblock Relative risk trees for censored survival data.
\newblock \emph{Biometrics}, 48, 1992{\natexlab{a}}.

\bibitem[M.Leblanch and J.Crowley(1992{\natexlab{b}})]{leblanch2}
M.Leblanch and J.Crowley.
\newblock Survival trees by goodness of split.
\newblock \emph{Journal of the {A}merican {S}tatistical {A}ssociation}, 88,
  1992{\natexlab{b}}.

\bibitem[Neal(1993)]{NealTRep}
R.M. Neal.
\newblock Probabilistic inference using {M}arkov chain {M}onte {C}arlo methods.
\newblock \emph{Technical report, Department of Computer Science, University of
  Toronto}, 1993.

\bibitem[Nott and Green(2004)]{Nott}
D.J. Nott and P.J. Green.
\newblock Bayesian veriable selection and the {S}wendsen-{W}ang algorithm.
\newblock \emph{Journal of {C}omputational and {G}raphical {S}tatistics},
  13:\penalty0 141--157, 2004.

\bibitem[Nummelin(1984)]{Nummelin}
E.~Nummelin.
\newblock \emph{General Irreducible Markov Chains on Non-Negative Operators}.
\newblock Cambridge University Press, 1984.

\bibitem[Pasetto et~al.(2003)Pasetto, E., and S.]{PasettoMonfardini}
L.M. Pasetto, Rossi E., and Monfardini S.
\newblock Liver metastases of colorectal cancer: medical treatment.
\newblock \emph{Anticancer}, 23:\penalty0 4245--56, 2003.

\bibitem[Peskun(1973)]{Peskun}
P.H. Peskun.
\newblock Optimum {M}onte {C}arlo sampling using {M}arkov chain.
\newblock \emph{Biometrika}, 60:\penalty0 607--612, 1973.

\bibitem[Pittman et~al.(2004)Pittman, Huang, Dressman, Horng, Cheng, Tsou,
  Chen, Bild, Iversen, Huang, Nevins, and West]{MikeJen}
J.~Pittman, E.~Huang, H.~Dressman, C.F. Horng, S.H. Cheng, M.H. Tsou, C.M.
  Chen, A.~Bild, E.S. Iversen, A.T. Huang, J.R. Nevins, and M.~West.
\newblock Integrated modeling of clinical and gene expression information for
  personalized prediction of disease outcomes.
\newblock \emph{Proceedings of the {N}ational {A}cademy of {S}ciences},
  101:\penalty0 8431--8436, 2004.

\bibitem[Raftery et~al.(1997)Raftery, Madigan, and Hoeting]{Raft97}
A.~E. Raftery, D.~Madigan, and J.~A. Hoeting.
\newblock Bayesian model averaging for linear regression models.
\newblock \emph{Journal of the {A}merican {S}tatistical {S}ociety},
  92:\penalty0 179--191, 1997.

\bibitem[Robert(1998)]{RobDiag}
C.P. Robert.
\newblock \emph{Discretization and MCMC convergence assessment}.
\newblock Springer, New York, 1998.

\bibitem[Robert and Casella(1999)]{RobCas}
C.P. Robert and G.~Casella.
\newblock \emph{Monte Carlo Statistical Methods}.
\newblock Springer, 1999.

\bibitem[Roberts et~al.(1997{\natexlab{a}})Roberts, Gelman, and Gilks]{RGG}
G.O. Roberts, A.~Gelman, and W.R. Gilks.
\newblock Weak convergence and optimal scaling of random walk {M}etropolis
  algorithms.
\newblock \emph{The {A}nnals of {A}pplied {P}robability}, 7:\penalty0 110--120,
  1997{\natexlab{a}}.

\bibitem[Roberts et~al.(1997{\natexlab{b}})Roberts, Gelman, and
  Gilks]{Roberts97}
G.O. Roberts, A.~Gelman, and W.R. Gilks.
\newblock Weak convergence and optimal scaling of random walk {M}etropolis
  algorithms.
\newblock \emph{The {A}nnals of {A}pplied {P}robability}, 7:\penalty0 110--120,
  1997{\natexlab{b}}.

\bibitem[Schwarz(1978)]{Schwarz}
G.~Schwarz.
\newblock Estimating the dimension of a model.
\newblock \emph{The {A}nnals of {S}tatistics}, 6:\penalty0 461--464, 1978.

\bibitem[Smith and Roberts(1993)]{SR}
A.~F.M. Smith and G.O. Roberts.
\newblock Bayesian computations via the {G}ibbs sampler and related {M}arkov
  chain {M}onte {C}arlo methods.
\newblock \emph{Journal of the {R}oyal {S}tatistical {S}ociety, series {B}},
  55:\penalty0 3--23, 1993.

\bibitem[Smith and Kohn(1996)]{SmithKohn}
M.~Smith and R.~Kohn.
\newblock Nonparametric regression using {B}ayesian veriable selection.
\newblock \emph{Journal of {E}conometrics}, 75:\penalty0 317--343, 1996.

\bibitem[Sun(1997)]{Sun}
Dongchu Sun.
\newblock A note on noninformative priors for {W}eibull distributions.
\newblock \emph{Journal of {S}tatistical {P}lanning and {I}nference},
  61:\penalty0 319--338, 1997.

\bibitem[Swedensen and Wang(1987)]{SWang}
R.H. Swedensen and J.S. Wang.
\newblock Nonuniversal critical dynamics in {M}onte {C}arlo simulations.
\newblock \emph{Physical {R}eview {L}etters}, 58:\penalty0 86--88, 1987.

\bibitem[Tanner and Wong(1987)]{TWong}
M.A. Tanner and W.H. Wong.
\newblock The calculation of posterior distributions via data augmentation.
\newblock \emph{Journal of the {A}merican {S}tatistical {A}ssociation},
  82:\penalty0 528--541, 1987.

\bibitem[Tierney(1994)]{Tierney}
L.~Tierney.
\newblock Markov chains for exploring posterior distributions.
\newblock \emph{The {A}nnals of {S}tatistics}, 22:\penalty0 1701--1728, 1994.

\end{thebibliography}
\bibliographystyle{plainnat}
\section*{Acknowledgements}
The author thanks Richard Gill, Antonietta Mira and Gareth Roberts for many helpful comments
on an earlier version of this manuscript and Mike West, who provided
the essential motivation for the development of the model in Section $5$.
The software implementing  the models and the MCMC algorithms employed
in Sections $3$, $4$ and $5$ can be obtained in the form of three \texttt{MATLAB}
modules upon request to the author.
\end{document}